\documentclass[10pt]{elsarticle}
\usepackage{amssymb,amsmath}
\usepackage{enumitem}
\usepackage{bbm}
\usepackage{mathtools}
\usepackage[margin=1in,papersize={8.5in,11in}]{geometry}
\usepackage[utf8]{inputenc}
\usepackage[english]{babel}
\usepackage{amsthm}
\usepackage{float}
\usepackage{amsfonts,mathrsfs}
\usepackage{bm}
\usepackage{makecell}
\usepackage{tikz}
\usetikzlibrary{matrix}

\newtheorem{hyp}{Hypothesis}
\usepackage{epstopdf}
\usepackage{subcaption}
\usepackage{tablefootnote}
\graphicspath{ {images/} }
\usepackage{comment}

\newtheorem{theorem}{Theorem}
\newtheorem{corollary}{Corollary}[theorem]
\newtheorem{lemma}[theorem]{Lemma}
\newtheorem{prop}{Proposition}
\theoremstyle{definition}
\newtheorem{defn}{Definition}

\def\R{\mathbb{R}}

\def\Pol{\text{Pol}}

\def\A{\mathcal{A}}

\def\S{\mathcal{S}} 
 
\def\M{\mathcal{M}}
  
\def\N{\mathcal{N}}

\def\P{\mathcal{P}}
\def\U{\mathcal{U}} 
\def\X{\mathcal{X}}
\def\Y{\mathcal{Y}}
\def\Re{\operatorname{Re}} 
\def\Im{\operatorname{Im}}  
\def\Q{\mathcal{Q}}

\def\C{\mathcal{C}}
\def\G{\mathcal{G}}
\def\K{\mathcal{K}} 
\def\F{\mathcal{F}} 
\def\I{\mathcal{I}}
\def\U{\mathcal{U}}

\def\Pi{\mathbf{P}}

\def\tP{\tilde{\mathcal{P}}}
\def\tQ{\tilde{\mathcal{Q}}}
\def\tK{\tilde{\mathcal{K}}}
\def\tS{\tilde{\mathcal{S}}}
\journal{arXiv}

\begin{document}
\begin{frontmatter}
\title{Hypoellipticity and the Mori-Zwanzig formulation of 
stochastic differential equations}


\author[ucm]{Yuanran Zhu}
\author[ucsc]{Daniele Venturi\corref{correspondingAuthor}}
\address[ucsc]{Department of Applied Mathematics, University of California Santa Cruz\\ Santa Cruz (CA) 95064}
\address[ucm]{Department of Applied Mathematics, University of California Merced\\ Merced (CA) 95343}
\cortext[correspondingAuthor]{Corresponding author}
\ead{venturi@ucsc.edu}
 
\begin{abstract}
We develop a thorough mathematical analysis of the effective Mori-Zwanzig (EMZ) equation governing the dynamics of noise-averaged observables in stochastic differential equations driven by multiplicative Gaussian white noise. Building upon recent work on hypoelliptic operators, we prove that the EMZ memory kernel and fluctuation terms converge exponentially fast in time to a unique equilibrium state which admits an explicit representation. We apply the new theoretical results to the Langevin dynamics of a high-dimensional particle system with smooth interaction potential.  
\end{abstract}
\end{frontmatter}

\section{Introduction}

The Mori-Zwanzig (MZ) formulation is a technique 
originally developed in statistical mechanics  \cite{Mori,zwanzig1961memory}
to formally integrate out phase variables in 
nonlinear dynamical systems by means of a 
projection operator. One of the main features of such formulation 
is that it allows us to systematically derive exact 
generalized Langevin equations (GLEs) 
\cite{zwanzig1973nonlinear,Chorin,Darve,Venturi_PRS} for 
quantities of interest, e.g., macroscopic observables, 
based on microscopic equations of motion. 
Such GLEs can be found in a variety of 
applications, including particle dynamics 
\cite{Li2015,VenturiBook,Yoshimoto2013,van1986brownian,
espanol1995statistical,espanol1995hydrodynamics}, 
fluid dynamics \cite{parish2017non,Falkena2019,parish2017dynamic}, 
and, more generally, systems described by nonlinear 
partial differential equations (PDEs)
\cite{venturi2014convolutionless,ChoPRS2014,
stinis2015renormalized,stinis2004stochastic,Stinis1,
Brennan2018,Venturi_JCP,Venturi_IJHMT,lu2017data,lin2019data}. 
Computing the solution to the MZ equation is usually 
a daunting task. One of the main 
difficulties is the approximation of the memory integral 
(convolution term) and the fluctuation term, which encode the 
interaction between the so-called orthogonal dynamics and the 
dynamics of the quantity of interest. 
The orthogonal dynamics is essentially a high-dimensional flow 
governed by an integro-differential equation that is hard to solve. 
The mathematical properties of the orthogonal dynamics, 
and hence the properties of the MZ memory integral and the MZ fluctuation term are not well understood.
Kupferman, Givon and Hald \cite{givon2005existence} 
proved existence and uniqueness of the orthogonal 
dynamics for deterministic dynamical systems and 
Mori's projection operators. 
More recently, we proved uniform boundedness 
of the orthogonal dynamics propagator 
for Hamiltonian systems using semigroup estimates \cite{zhu2018estimation,engel1999one}.

The main objective of this paper is to generalize
the MZ formulation to stochastic differential equations 
(SDEs) \cite{espanol1995hydrodynamics,hudson2018coarse} 
driven by multiplicative Gaussian white noise. 
In particular, we aim at developing a 
thorough mathematical analysis of the 
so-called effective Mori-Zwanzig (EMZ) equation 
governing the dynamics of noise-averaged observables, i.e., 
smooth functions of the stochastic flow
generated by the SDE which are averaged over the 
probability measure of the random noise. 
To this end,  we build upon recent work of Eckmann \& Hairer 
\cite{eckmann1999non, eckmann2000non,eckmann2003spectral},
H\`erau \& Nier \cite{herau2004isotropic} and Helffer \& 
Nier \cite{nier2005hypoelliptic} on the spectral properties 
of backward Kolmogorov operators, and show that 
the generator of EMZ orthogonal dynamics 
has a discrete spectrum that lies within cusp-shaped 
region of the complex plane. This allows us to 
rigorously prove exponential relaxation to a unique 
equilibrium state for both the EMZ memory kernel 
and the EMZ fluctuation term.

This paper is organized as follows. In Section \ref{sec:MZ_time}, 
we develop a self-consistent MZ formulation for 
stochastic differential equations driven by multiplicative 
Gaussian white noise and derive the 
effective Mori Zwanzig  equation governing 
the dynamics of noise-averaged observables. 
In Section \ref{sec:abs_ana} we study the 
theoretical properties of the EMZ equation. 
To this end, we first review H\"ormander's theory of linear 
hypoelliptic operators, and then show how such theory  can 
used to prove exponential convergence of the 
EMZ orthogonal dynamics propagator to a unique 
equilibrium state.
In Section \ref{sec:app}, we apply our theoretical 
results to the Langevin dynamics of 
high-dimensional particle systems with smooth interaction 
potentials that grow at most polynomially fast at infinity. The 
main findings are summarized in Section \ref{sec:conclusion}.

\section{The Mori-Zwanzig formulation of stochastic differential equations}
\label{sec:MZ_time}
Let us consider a $d$-dimensional stochastic 
differential equation on a smooth manifold $\M$
\begin{align}\label{eqn:sde}
\frac{d\bm x(t)}{dt}=\bm F(\bm x(t))+\bm \sigma(\bm x(t))\bm \xi(t), \qquad \bm x(0)=\bm x_0\sim \rho_0(x),
\end{align}
where $\bm F:\M\mapsto \R^d$ and 
$\bm \sigma: \M\rightarrow \R^{d\times m}$ are 
smooth functions, $\bm \xi(t)$ is 
$m$-dimensional Gaussian white noise with 
independent components, and $\bm x_0$ is a 
random initial state characterized in terms of 
a probability density function $\rho_0(\bm x)$.
The solution \eqref{eqn:sde} is 
a $d$-dimensional stochastic (Brownian) flow
on the manifold $\M$ \cite{kunita1997stochastic}.
As is well known, if $\bm F:\M\mapsto \R^d$ and 
$\bm \sigma: \M\rightarrow \R^{d\times m}$ are of class
$\mathcal{C}^{k+1}$ ($k\geq 0$) 
with uniformly bounded derivatives, then the solution 
to \eqref{eqn:sde} is global, and that the 
corresponding flow is a stochastic 
flow of diffemorphisms of class $\mathcal{C}^{k}$ (see \cite{Budhiraja2010,wihstutz2012diffusion,Watanabe1984}). 
This means that the stochastic flow is differentiable 
$k$ times (with continuous derivative), with respect 
to the initial condition for all $t$. 
Define the vector-valued phase space function 
(quantity of interest)
\begin{equation}
\begin{array}{cc}
\bm u\colon& \mathcal{M} \to \mathbb{R}^M \\
&\bm x  \mapsto \bm u(\bm x)
\end{array}
\qquad \text{(quantity of interest)}.
\label{observable}
\end{equation}

\noindent
{By evaluating $\bm u(\bm x)$ 
along the stochastic flow generated by 
the SDE \eqref{eqn:sde} and averaging  
over the Gaussian white noise we obtain} 
\begin{align}
\mathbb{E}_{\bm \xi(t)}[\bm u(\bm x(t))|\bm x_0]= 
\F(t,0)\bm u(\bm x_0).
\label{MarkovSemi}
\end{align}
The evolution operator $\F(t,0)$ is a Markovian 
semigroup \cite{da1996ergodicity}
generated by the following (backward) Kolmogorov 
operator \cite{Risken,kloeden2013numerical,venturi2018numerical}
\begin{align}
\K(\bm x_0)&=\sum_{k=1}^dF_k(\bm x_0)\frac{\partial}{\partial x_{0k}}
+\frac{1}{2}\sum_{j=1}^m\sum_{i,k=1}^d\sigma_{ij}(\bm x_0)
\sigma_{kj}(\bm x_0)\frac{\partial}{\partial x_{0i}\partial x_{0k}},
\label{KI}
\end{align}
{which corresponds to the It\^o interpretation of the SDE \eqref{eqn:sde}.}
Formally, we will write
\begin{equation}
\F(t,0)=e^{t\K}.
\label{eK}
\end{equation}
{To derive the 
effective Mori-Zwanzig (EMZ) equation 
governing the time evolution of the averaged
observable \eqref{MarkovSemi}, 
we introduce a projection operator $\P$ 
and the complementary projection $\Q=\I-\P$.
By following the formal procedure outlined in 
\cite{zhu2019generalized,zhu2018faber,dominy2017duality} 
we obtain}
\begin{equation}
\frac{\partial}{\partial t}e^{t\K}\bm u(0)
=e^{t\K}\mathcal{PK}\bm u(0)
+e^{t\Q\K\Q}\mathcal{QK}\bm u(0)+\int_0^te^{s\K}\P\K
e^{(t-s)\Q\K\Q}\mathcal{QK}\bm u(0)ds.\label{eqn:EMZ_full}
\end{equation}
Applying the projection operator 
$\P$ to \eqref{eqn:EMZ_full}
yields
\begin{equation} 
\frac{\partial}{\partial t}\mathcal{P}e^{t\K}\bm u(0)
=\mathcal{P}e^{t\K}\mathcal{PK}\bm u(0)
+\int_0^t\P e^{s\K}\mathcal{PK}
e^{(t-s)\Q\K\Q}\mathcal{QK}\bm u(0)ds.\label{eqn:EMZ_projected}
\end{equation}
Note that both the EMZ equation \eqref{eqn:EMZ_full} and its projected form \eqref{eqn:EMZ_projected} have the same structure as the classical MZ equation for deterministic (autonomous) systems 
\cite{zhu2019generalized,zhu2018estimation,zhu2018faber}.

\subsection{EMZ equation with Mori's projection operator}
Let us consider the weighted Hilbert space 
$H=L^2(\M,\rho)$, where $\rho$ is a positive weight 
function in $\M$. For instance, $\rho$ can be the 
probability density function of the random initial state 
$\bm x_0$. Let
\begin{equation}
\langle h,g\rangle_{\rho}=\int_{\M} 
h(\bm x)g(\bm x)\rho(\bm x)d\bm x 
\qquad h,g\in H
\label{ip}
\end{equation}
be the inner product in $H$. 
We introduce the following projection operator
\begin{align}
\label{Mori_P}
\P h=\sum_{i,j=1}^M G^{-1}_{ij}
\langle u_i(0),h\rangle_{\rho}u_j(0),
\qquad h\in H,
\end{align}
where $G_{ij}=\langle u_i(0),u_j(0)\rangle_{\rho}$
and $u_i(0)=u_i(\bm x)$ ($i=1,...,M$) are
$M$ linearly independent functions. With $\P$ 
defined as in \eqref{Mori_P}, we can write 
the EMZ equation \eqref{eqn:EMZ_full} and its 
projected version \eqref{eqn:EMZ_projected}
as
\begin{align}
\frac{d\bm q(t)}{dt} &= \bm \Omega \bm q(t) +
\int_{0}^{t}\bm K(t-s)\bm q(s)ds+\bm f(t),\label{gle_full}\\
\frac{d}{dt}\P{\bm q}(t) &= \bm \Omega\P {\bm q}(t)+ 
\int_{0}^{t} \bm K(t-s) \P {\bm q}(s)ds,\label{gle_projected}
\end{align}
{where $\bm q(t)=\mathbb{E}
[\bm u(\bm x(t))|\bm x_0]$ (column vector)} and 
\begin{subequations}
\begin{align}
		G_{ij} & = \langle u_{i}(0), u_{j}(0)\rangle_{\rho}
		\quad \text{(Gram matrix)},\label{gram}\\
		\Omega_{ij} &= \sum_{k=1}^MG^{-1}_{jk}
		\langle u_{k}(0), \K u_{i}(0)\rangle_{{\rho}}\quad 
		\text{(streaming matrix)},\label{streaming}\\
		K_{ij}(t) & =\sum_{k=1}^M G^{-1}_{jk}
		\langle u_{k}(0), \K e^{t\Q\K\Q}\Q\K u_{i}(0)\rangle_{\rho}\quad 
		\text{(memory kernel)},\label{SFD}\\
		 f_i(t)& =e^{t\Q\K\Q}\Q\K u_i(0) \quad 
		\text{(fluctuation term)}.\label{f}
	\end{align}
\end{subequations}
\noindent
{In equations \eqref{gram}-\eqref{f} 
we have $u_j(0)= q_j(0)=u_j(\bm x_0)$ ($j=1,\ldots,M)$.} Also,  
the Kolmogorov operator $\K$ is not skew-symmetric 
relative to $\langle,\rangle_{\rho}$ and therefore it is 
not possible (in general) to represent the memory 
kernel ${\bm K}(t)$ as a function of the auto-correlation 
of $\bm f(t)$ using the second fluctuation-dissipation 
theorem \cite{zhu2019generalized,zhu2021generalized}.

\subsection{An Example: EMZ formulation of the Ornstein-Uhlenbeck SDE} 

Let us consider the Ornstein-Uhlenbeck process defined by 
the solution to the It\^{o} stochastic differential equation
{\begin{align}\label{eqn:OUsde}
\frac{d}{dt}x=-\theta x+\sigma\xi(t),
\end{align} 
where $\sigma$ and $\theta$ are positive parameters}
and $\xi(t)$ is Gaussian white noise with correlation function 
$\langle\xi(t),\xi(s)\rangle=\delta(t-s)$. As is well-known, the 
Ornstein-Uhlenbeck process is ergodic and it 
admits a stationary (equilibrium) Gaussian 
distribution {$\rho_{eq}=\N(0, \sigma^2/2\theta)$}. 
Let $x(0)$ be a random initial state with probability 
density function {$\rho_0=\rho_{eq}$}.
%
The conditional mean and conditional covariance 
function of the process $x(t)$ are given by
{\begin{align}
\mathbb{E}_{\xi(t)}[x(t)|x(0)]=&x(0)e^{-\theta t},
\label{Ou_mt}\\
\mathbb{E}_{\xi(t)}[x(t)x(s)|x(0)]=&
x(0)^2e^{-\theta(t+s)}+\frac{\sigma^2}{2\theta}
\left(e^{-\theta|t-s|}-e^{-\theta(t+s)}
\right).
\end{align}}
Averaging over the random initial state yields 
{
\begin{align}
\mathbb{E}_{x(0)}[\mathbb{E}_{\xi(t)} [x(t)|x(0)]]&=0,\label{Ou_Mt}\\
\mathbb{E}_{x(0)}[\mathbb{E}_{\xi(t)} [x(t)x(s)|x(0)]]
&=\frac{\sigma^2}{2\theta}e^{-\theta|t-s|}.\label{Ou_Ct}
\end{align}}
At this point, we define the projection operators 
{
\begin{align}\label{EMZ_moriP}
\P_1(\cdot)=\frac{\langle \cdot\rangle_{\rho_{eq}}}{\langle x(0)\rangle_{\rho_{eq}}}x(0), \qquad \P_2(\cdot) =\frac{\langle \cdot,x(0)\rangle_{\rho_{eq}}}{\langle x(0),x(0)\rangle_{\rho_{eq}}}x(0).
\end{align}}
The Kolmogorov backward operator associated with \eqref{eqn:OUsde} is
{
\begin{align}
\K(x)=-\theta x\frac{\partial}{\partial x}+\frac{\sigma^2}{2}\frac{\partial^2}{\partial x^2}.
\end{align}}
By using the identity 
\begin{align}
\frac{d}{dt}\mathbb{E}_{ x(0)}[\mathbb{E}_{\xi(t)}[x(t)x(0)|x(0)]]
=\frac{d}{dt}\mathbb{E}_{ x(0)}[\mathbb{E}_{\xi(t)}[x(t)|x(0)]x(0)]
=\left\langle\frac{d}{dt}\F(t,0)x(0), x(0)\right\rangle_{\rho_{eq}},
\end{align}
it is straightforward to verify that the EMZ equation \eqref{gle_full} with $\P=\P_1$, 
and the EMZ equation \eqref{gle_projected} with $\P=\P_2$ can be written, respectively, as 
\begin{align}
\frac{d}{dt}M(t)=-\theta M(t),\qquad 
\frac{d}{dt}C(t)=-\theta C(t).
\label{OU_mt_GLE}
\end{align}
{Here, $ M(t)=\mathbb{E}_{\xi(t)}[x(t)|x(0)]$ is the 
conditional mean of $x(t)$, while 
$C(t)=\mathbb{E}_{x(0)}[\mathbb{E}_{\xi(t)} [x(t)x(0)|x(0)]]$ is autocorrelation
function of $x(t)$.}
Clearly, equations \eqref{OU_mt_GLE} are the exact evolution equations governing 
$M(t)$ and $C(t)$ . In fact, their solutions coincide with \eqref{Ou_mt} and \eqref{Ou_Ct}, respectively. {Note that $M(t)$ is a stochastic 
process ($x(0)$ is random), while $C(t)$ is a deterministic function.}

\section{Analysis of the effective Mori-Zwanzig equation}
\label{sec:abs_ana}
In this section we develop an in-depth mathematical 
analysis of the effective Mori-Zwanzig equation
\eqref{eqn:EMZ_full} using H\"ormander's
theory \cite{herau2004isotropic,nier2005hypoelliptic,
villani2009hypocoercivity,ottobre2011exponential}. 
In particular, we build upon the result of H\'erau and 
Nier \cite{herau2004isotropic}, Eckmann and Hairer 
\cite{eckmann1999non,eckmann2000non,
eckmann2003spectral}, and Helffer and Nier 
\cite{nier2005hypoelliptic} on linear hypoelliptic 
operators to prove that the generator of the 
EMZ orthogonal dynamics, i.e., $\Q\K\Q$, satisfies a hypoelliptic estimate.
Consequently, the propagator $e^{t\Q\K\Q}$ 
converges exponentially fast (in time) to statistical equilibrium. 
This implies that both the EMZ memory kernel \eqref{SFD} 
and fluctuation term \eqref{f} converge exponentially 
fast to an equilibrium state. 
One of the key results of such analysis is the fact 
that the spectrum of $\Q\K\Q$ lies within a cusp-shaped 
region of the complex half-plane. 
For consistency with the literature on hypoelliptic operators, 
we will use the negative of $\K$ and $\Q\K\Q$ as semigroup 
generators and write the 
semigroups appearing in EMZ equation \eqref{eqn:EMZ_full} 
as $e^{-t\K}$ and $e^{-t\Q\K\Q}$. Clearly,
if $\K$ and $\Q\K\Q$ are dissipative then 
$-\K$ and $-\Q\K\Q$ are accretive.
Unless otherwise stated, throughout 
this section we consider scalar quantities of 
interest, i.e., we set $M=1$ in equation \eqref{observable}.

\subsection{Analysis of the Kolmogorov operator}
\label{sec:K}

The Kolmogorov operator \eqref{KI} is a H\"ormander-type operator 
which can be written in the general form
\begin{align}\label{general_K}
\K(\bm x)=\sum_{i=1}^m \X_i^*(\bm x)\X_i(\bm x)+\X_0(\bm x)+f(\bm x),
\end{align} 
where $\X_i(\bm x)$ ($0\leq i\leq m$) denotes a 
{first-order partial differential operator} in the variable $x_i$ 
with space-dependent coefficients, $\X_i^*(\bm x)$ 
is the formal adjoint of $\X_i(\bm x)$ in $L^2(\R^n)$, 
and $f(\bm x)$ is a function that has at most polynomial 
growth at infinity.
To derive useful spectral estimates for $\K$, it 
is convenient to first provide some definitions 
\begin{defn}
\label{Pol_cond} Let $N$ be a real number. Define
\begin{align*}
\Pol_0^N=\left\{f\in\C^{\infty}(\R^n)\,\mid\sup_{\bm x \in\R^n}(1+\|\bm x\|)^{-N}|\partial^{\bm \alpha}f(\bm x)|\leq C_{\bm \alpha}\right\},
\end{align*}
where $\bm \alpha$ is a multi-index of arbitrary 
order. Note that $\Pol_0^N$ is the set of infinitely 
differentiable functions growing at most polynomially 
as $\|\bm x\|\rightarrow\infty$.
Similarly, we define the space of $k$-th order 
differential operators with coefficients growing 
at most polynomially with $\bm x$ as 
\begin{align*}
\Pol_k^N= \left\{ \G:\C^{\infty}(\R^n)\rightarrow \C^{\infty}(\R^n)\,\mid 
\G(\bm x)=G_0(\bm x)+\sum_{j=1}^n\sum_{i=1}^kG_j^i(\bm x)\partial_j^i,\quad G_j^i\in \Pol_0^N\right\}.
\end{align*}
\end{defn}  
\noindent
It is straightforward to verify that if $\X\in\Pol_k^N$ and $\Y\in\Pol_l^M$ 
then the operator commutator $[\X,\Y]=\X\Y-\Y\X$ is in 
$\Pol_{k+l-1}^{N+M}$. 

\begin{defn}
\label{nond_cond}
The family of operators $\{\A_1,\ldots,\A_m\}$ defined as 
\begin{equation}
\A_i(\bm x)=\sum_{j=1}^nA_{ij}(\bm x)\partial_j\qquad i=1,...,m
\end{equation}
is called {\em non-degenerate} if there are 
two constants $N$ and $C$ such that 
\begin{align*}
\|\bm y\|^2\leq C(1+\|\bm x\|^2)^N\sum_{i=1}^m\langle 
\A_i(\bm x),\bm y \rangle^2 \qquad \forall \bm x,\bm y\in\R^n,
\end{align*}
where $\displaystyle \langle \A_i(\bm x),\bm y\rangle=
\sum_{j=1}^nA_{ij}(\bm x)y_j$.
\end{defn}
\noindent
It was recently shown by Eckmann and Hairer 
\cite{eckmann2000non,eckmann2003spectral} that 
$\K$ is hypoelliptic if the Lie algebra 
generated by the operators $\{\X_0,\ldots,\X_m\}$ 
in \eqref{general_K} is non-degenerate. 
The main result can be summarized as follows:
\begin{prop}[Eckmann and Hairer \cite{eckmann2003spectral}]\label{Lie_cond}
Let $\{\X_0,\ldots,\X_m\}$ and $f$ in  \eqref{general_K} 
satisfy the following conditions:
\begin{enumerate}
\item $\X_j\in\Pol_1^N $ for all $j=0,\ldots,m$, and $f\in\Pol_0^N$;
\label{hypo_cond1}
\item There exits a finite integer $M$ such that 
the family of operators consisting of $\{\X_i\}_{i=0}^m$, 
$\{[\X_i,\X_j]\}_{i,j=1}^m$, $\{[\X_i,[\X_j,\X_k]]\}_{i,j,k=1}^m$ 
and so on up to the commutators of rank $M$ is non-degenerate;
\label{hypo_cond2}
\end{enumerate} 
Then the operator  $\K$ defined in \eqref{general_K} and 
$\partial_t+ \K$ are both hypoelliptic. 
\end{prop} 

\noindent 
Conditions \ref{hypo_cond1} and \ref{hypo_cond2}
in Proposition \ref{Lie_cond} are called 
{\em poly-H\"{o}rmander conditions}. 
Eckmann et al. \cite{eckmann2000non,eckmann1999non} 
also proved the hypoellipticity of the operator 
$\partial_t+\K^*$ for a specific heat condition model, 
which guarantees smoothness (in time) of the transition 
probability governed by the Kolmogorov forward
equation.
Hereafter we review additional important 
properties of the Kolmogorov operator $\K$. 
As a differential operator with $C^{\infty}$ 
tempered  coefficients (i.e. with all derivatives 
polynomially bounded), $\K$ and its formal adjoint 
$\K^*$ are defined in the Schwartz space 
$\mathscr{S}(\R^{n})$, which is dense 
in $L^p(\R^n)$ ($1\leq p<\infty$). 
On the other hand, since $\K$ and $\K^*$ 
are both closable operators, all estimates we obtain 
in this Section hold naturally in $\mathscr{S}(\R^n)$, 
which can be extended to $L^2(\R^n)$. Hence, we 
do need to distinguish between $\K$ and its closed 
extension in $L^2(\R^n)$.
We now introduce a family of weighted 
Sobolev spaces 
\begin{align}
\S^{\alpha,\beta}=\{u\in \mathscr{S}'(\R^n):\,
\Lambda^{\alpha}\bar{\Lambda}^{\beta}  
u\in L^2(\mathbb{R}^n)\quad \alpha,\beta\in\R\},
\label{wss}
\end{align}
where $\mathscr{S}'(\R^n)$ the space of tempered 
distributions in $\R^n$. The operator $\bar{\Lambda}^{\beta}$ is product operator defined as $\bar{\Lambda}^{\beta}:=(1+\|\bm x\|^2)^{\beta/2}$, while  $\Lambda^{\alpha}$
is a pseudo-differential operator (see
\cite{eckmann1999non,eckmann2003spectral,
herau2004isotropic}) 
that reduces to 
\begin{equation}
\Lambda^2=1-\Delta
\end{equation} 
for $\alpha=2$. 
The weighted Sobolev space \eqref{wss} is equipped 
with the scalar product 
\begin{align*}
\langle  h, g \rangle_{\alpha,\beta}=\langle \Lambda^{\alpha}
\bar{\Lambda}^{\beta} h,\Lambda^{\alpha}\bar{\Lambda}^{\beta} g
\rangle_{L^2}, 
\end{align*}
which induces the Sobolev norm $\|\cdot\|_{\alpha,\beta}$. 
Throughout the paper $\|\cdot\|$ denotes the standard $L^2$ norm. With the above definitions it is possible to prove the following 
important estimate on spectrum of the Kolmogorov operator 
$\K$.
\begin{theorem}[Eckmann and Hairer \cite{eckmann2003spectral}]\label{EH_spec_est}
Let $\K\in\Pol_2^N$ be an operator of the form 
\eqref{general_K} satisfying conditions 1. and 2. 
in Proposition \ref{Lie_cond}. Suppose that the 
closure of $\K$ is a maximal-accretive operator in 
$L^2(\R^n)$ and that for every $\epsilon>0$ there 
are two constants $\delta>0$ and $C>0$ such that
\begin{align}\label{1epsilon_con}
\| u\|_{\delta,\delta}\leq C(\| u\|_{0,\epsilon}+\|\K u\|)
\end{align}
for all $u\in\mathscr{S}(\R^n)$. If, in addition, there 
exist two constants $\delta>0$ and $D>0$ such that 
\begin{align}\label{0epsilon_con}
\|u\|_{0,\epsilon}\leq D(\| u\|+\|\K u\|)
\end{align}
then $\K$ has compact resolvent when considered as an operator 
acting on $L^2(\R^n)$, whose spectrum $\sigma(\K)$ 
is contained in the following cusp-shaped region $\S_{\K}$ of the 
complex plane(see Figure \ref{fig:cusp}):
\begin{align}\label{formula_cusp}
\S_{\K}=\{z\in \mathbb{C}\,:\, \Re z\geq 0,\,\,
|z+1|<(8C_1)^{M/2}(1+\Re z)^{M}\}
\end{align}
for some positive constant $C_1$ and $M\in \mathbb{N}$.
\end{theorem}
\noindent
We remark that in 
\cite{eckmann2003spectral} the cusp 
$\S_{\K}$ is defined as $\S_{\K}=\{z\in 
\mathbb{C}\,:\, \Re z\geq 0,\,\,|\Im z|<
(8C_1)^{M/2}(1+\Re z)^{M}\}$. Clearly, $\S_{\K}$ 
in equation \eqref{formula_cusp} is also a valid cusp 
since it can be derived directly from \eqref{est} 
(see, e.g., the proof of Theorem 4.3 in 
\cite{eckmann2003spectral}).
\begin{figure}[t]
\centerline{
\includegraphics[height=8cm]{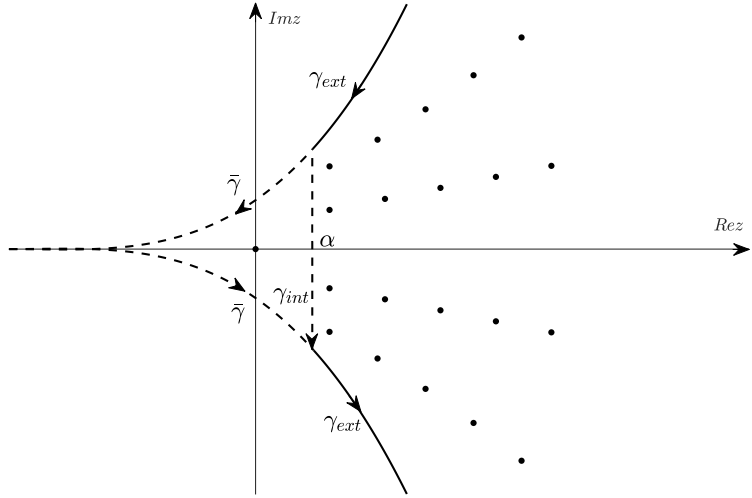}}
\caption{Sketch of the cusp-shaped region of the complex 
plane  enclosing the spectrum of $\K$ and $\Q\K\Q$. The curve 
$\gamma_{ext}$ represents the boundary 
of the cusp. We denote by $\partial\S_{\K}$ the 
curve defined by the union of $\gamma_{ext}$ and $\bar{\gamma}$, 
while $\partial\S'_{\K}$ is defined by the union $\gamma_{ext}$ 
(up to intersection with $\gamma_{int}$) and $\gamma_{int}$.}
\label{fig:cusp} 
\end{figure}
\noindent
One of the key estimates used by 
Eckmann and Hairer in the proof of Theorem \ref{EH_spec_est}
is
\begin{align}\label{est}
\frac{1}{4}|z+1|^{2/M}\|u\|^2\leq C_1\left([1+\Re z]^2\|u\|^2+
\|(\K-z) u\|^2\right), \qquad \forall \Re z\geq 0.
\end{align} 
In a series of papers, H\`erau, Nier and Helffer 
\cite{herau2004isotropic,nier2005hypoelliptic} 
proved that the Kolmogorov operator $\K$ 
corresponding to classical Langevin dynamics 
generates a semigroup $e^{-t\K}$ that decays 
exponentially fast to an equilibrium state. 
Hereafter we show that similar results can be 
obtained for Kolmogorov operators 
in the more general form \eqref{general_K}. 
\begin{theorem}\label{exp_est_K}
Suppose that $\K$ satisfies all conditions in 
Theorem \ref{EH_spec_est}. If the spectrum $\sigma(\K)$ 
of $\K$ in $L^2(\R^n)$ is such that 
\begin{align}\label{K_spec_0}
\sigma(\K)\cap i\R=\{0\},
\end{align} 
then for any $0<\alpha<\min (\Re \sigma(\K)/\{0\})$, 
there exits a positive constant $C=C(\alpha)$ such 
that the estimate
\begin{align}\label{K_estimation}
\|e^{-t\K} u_0-\pi_0u_0\|\leq Ce^{-\alpha t}\|u_0\|
\end{align}
holds for all $ u_0\in L^2(\R^n)$ and for all $t>0$, where 
$\pi_0$ is the spectral projection onto the kernel of $\K$. 
\end{theorem}
\begin{proof}
The Kolmogorov operator $\K$ is closed, maximal-accretive 
and densely defined in $L^2(\R^n)$.  Hence, by the  
Lumer-Phillips theorem, the semigroup $e^{-t\K}$ is a contraction
in $L^2(\R^n)$. 
It was shown in \cite{eckmann2003spectral,herau2004isotropic} 
that the core of $\K$ is the Schwartz space, and that 
the hypoelliptic estimate \eqref{est} holds 
for any $u\in L^2(\R^n)$. 
According to Theorem \ref{EH_spec_est}, $\K$ only has a 
discrete spectrum, i.e., $\sigma(\K)=\sigma_{dis}(\K)$. 
Condition \eqref{K_spec_0} requires that $\lambda=0$
is the only eigenvalue on the immaginary axis $i\R$. 
This condition, together with the von-Neumann theorem (see the proof of Theorem 6.1 
in \cite{nier2005hypoelliptic}), allows us to obtain 
a weakly convergent Dunford integral \cite{reed1975ii} representation 
of the semigroup $e^{-t\K}$ given by 
\begin{align}\label{Dunford_integal}
e^{-t\K}u_0-\pi_0u_0=\frac{1}{2\pi i}
\int_{\partial \S'_{\K}}e^{-tz}(z-\K)^{-1} u_0dz,
\end{align}
where $\partial \S'_{\K}=\gamma_{int}\, \cup\, \gamma_{ext}$ 
is the union of the two curves shown in Figure \ref{fig:cusp}, 
and $(z-\K)^{-1}$ is the resolvent of $\K$. 
Weak convergence is 
relative to the inner product
\begin{equation}
\langle (e^{-t\K}-\pi_0)u_0,\phi\rangle = \frac{1}{2\pi i}
\int_{\partial \S'_{\K}}\langle e^{-tz}(z-\K)^{-1} u_0,\phi\rangle dz
\end{equation}
for $u_0\in L^2(\R^n)$ and $\phi\in D(\K^*)$.
Equation \eqref{Dunford_integal} allows 
us to formulate the semigroup estimation problem as an 
estimation problem involving an integral in the complex plane. 
In particular, to derive the upper bound \eqref{K_estimation}, 
we just need an upper bound for the norm of resolvent $(z-\K)^{-1}$. 
To derive such bound, we notice that for all 
$z\not\in\S_{\K}$, where $\S_{\K}$ is the cusp \eqref{formula_cusp}, 
and $\Re z\geq 0$, we have $|z+1|^{2/M}\geq (8C_1)(1+\Re z)^2$.  
A substitution of this inequality into \eqref{est} yields, for all 
$u\in L^2(\R^n)$
\begin{align*}
\frac{1}{8}|z+1|^{2/M}\| u\|^2&\leq C_1\|(\K-z)u\|^2, \qquad 
\forall \Re z\geq 0,z\not\in\S_{\K}.
\end{align*}
Hence,  $\|(\K-z)^{-1}\|\leq \sqrt{8C_1}|z+1|^{-1/M}$. 
Next, we rewrite the Dunford integral \eqref{Dunford_integal} as 
\begin{align}
\frac{1}{2\pi i}\int_{\partial \S'_{\K}}e^{-tz}(z-\K)^{-1}u_0dz=
\frac{1}{2\pi i}\int_{\gamma_{int}}e^{-tz}(z-\K)^{-1}u_0dz+
\frac{1}{2\pi i}\int_{\gamma_{ext}}e^{-tz}(z-\K)^{-1}u_0dz.
\label{Dunf1}
\end{align}
Since $(\K-z)^{-1}$ is a compact linear operator, 
we have that for any $0<\alpha<\min (\Re \sigma(\K)/\{0\})$ 
there exits a constant $C_{\alpha}>0$ such that 
$\|(\K-\alpha) u\|\geq C_{\alpha}\|u\|$. On the other hand, 
$\K$ is also a real operator, which implies that 
for all complex numbers $z=(\alpha+iy)\not\in\sigma(\K)$, 
we have 
\begin{align} 
\|(\K-(\alpha+iy)) u\|^2&=\|(\K-\alpha) u\|^2+
y^2\|u\|^2\geq (C_{\alpha}^2+y^2)\| u\|^2,\nonumber
\end{align}
i.e., 
\begin{align}
\|(\K-(\alpha+iy))^{-1} u\|&\leq\frac{1}{\sqrt{C_{\alpha}^2+y^2}}\|u\|\label{for:Reso_bound}.
\end{align}
This suggests that the resolvent $(\K-z)^{-1}$ is uniformly 
bounded by $1/C_{\alpha}$ along the line $\gamma_{int}$, which leads to 
\begin{align}\label{est1}
\left\|\frac{1}{2\pi i}\int_{\gamma_{int}}
e^{-tz}(z-\K)^{-1}u_0dz\right\|\leq Ce^{-\alpha t}\|u_0\|.
\end{align}
The boundary $\gamma_{ext}$ is defined by all complex numbers 
$z=x+iy$  such that $|z+1|=(8C_1)^{M/2}(1+\Re z)^{M}$. 
Also, if $z\not\in\S_{\K}$ then the norm 
of the resolvent is bounded by $\|(\K-z)^{-1}\|\leq \sqrt{8C_1}|z+1|^{-1/M}=(x+1)^{-1}$. Combining these two inequalities 
yields
\begin{align}\label{est2}
\left\|\frac{1}{2\pi i}\int_{\gamma_{ext}}
e^{-tz}(z-\K)^{-1}u_0dz\right\|
&\leq C\|u_0\|\int_{\gamma_{ext}}^{\infty}
e^{-tx}(1+x)^{-1}dz\nonumber\\
&\leq C\|u_0\|\int_{\alpha}^{\infty}e^{-tx}dx\leq
C\|u_0\|\frac{e^{-\alpha t}}{t}\qquad t>0.
\end{align}
At this point we recall that $e^{-t\K}$ is a dissipative 
semigroup and that $\pi_0$ is a projection operator 
into the kernel of $\K$. This allows us to write 
$\|e^{-t\K}u_0-\pi_0u_0\|=
\|e^{-t\K}(u_0-\pi_0u_0)\|\leq \|u_0-\pi_0u_0\|$.
By combining this inequality with \eqref{Dunford_integal}, \eqref{Dunf1}, 
\eqref{est1} and \eqref{est2} we see that there exists 
a constant $C=C(\alpha)$ such that 
\begin{align}\label{Taylor_K}
\|e^{-t\K} u_0-\pi_0u_0\|\leq Ce^{-\alpha t}\| u_0\|.
\end{align}
This completes the proof.

\end{proof}

\noindent
In the following Corollary we derive an upper 
bound for the norm of the derivatives of the 
semigroup $e^{-t\K}$.
\begin{corollary}
\label{Taylor_K_conv}
Suppose that $\K$ satisfies all conditions listed in 
Theorem \ref{exp_est_K}. Then for any $t> 0$, 
the $n$-th order time derivative of the 
semigroup $e^{-t\K}$ satisfies
\begin{align}\label{power_K}
\|e^{-t\K}\K^nu_0\|\leq\left(\|\pi_0\K\|+B\left(\frac{t}{n}\right)\right)^n\|u_0\|,
\end{align}
where
\begin{equation}
 B(t)= Ce^{-\alpha t}\left[1+\frac{1}{t}+\frac{1}{t^2}+
\cdots\frac{1}{t^{M}}\right],
\label{Bt}
\end{equation}
$C$ is a positive constant,  $\alpha$ and $\pi_0$ 
are defined in Theorem \ref{exp_est_K}, and $M$ is the constant 
defining the cusp \eqref{formula_cusp}.
\end{corollary}

\begin{proof}
By combining the resolvent identity $(z-\K)^{-1}\K=z(z-\K)^{-1}-\I$
with the Cauchy integral representation theorem and 
the Dunford integral representation \eqref{Dunford_integal} 
we obtain, for all $t>0$,
\begin{align}
e^{-t\K}\K-\pi_0\K 
&=\frac{1}{2\pi i}\int_{\partial \S'_{\K}}
e^{-tz}(z-\K)^{-1}\K  dz\nonumber\\
&=\frac{1}{2\pi i}\int_{\partial \S'_{\K}}ze^{-tz}(z-\K)^{-1} 
dz.
\label{dunford_power}
\end{align} 
As before, we split the integral along $\partial \S'_{\K}$ into 
the sum of two integrals (see equation \eqref{Dunf1})
\begin{align}
e^{-t\K}\K-\pi_0\K
=\frac{1}{2\pi i}\left(\int_{\gamma_{int}}ze^{-tz}(z-\K)^{-1}dz+\int_{\gamma_{ext}}ze^{-tz}(z-\K)^{-1} dz\right).
\label{dunford_power2}
\end{align} 
If $z=x+iy$ is in $\gamma_{int}$ then we have that 
$|z|$ is bounded by constant. By using the uniform 
boundedness of the resolvent \eqref{for:Reso_bound} 
we obtain
\begin{align}\label{esti1}
\left\|\frac{1}{2\pi i}\int_{\gamma_{int}}e^{-tz}z(z-\K)^{-1}
dz\right\|\leq Ce^{-\alpha t}.
\end{align}
To derive an upper bound for the second integral in
\eqref{dunford_power2}, we notice that if $z=x+iy$ is in 
$\gamma_{ext}$, then $|z|<|z+1|=(8C_1)^{M/2}(1+x)^M$ 
and $\|(z-\K)^{-1}\|\leq\sqrt{8C_1}|z+1|^{-1/M}=(1+x)^{-1}$. 
A substitution of these estimates into the second integral 
at the right hand side of \eqref{dunford_power2} yields, 
for all $t>0$, 
\begin{align}\label{esti2}
\left\|\frac{1}{2\pi i}\int_{\gamma_{ext}}ze^{-tz}
(z-\K)^{-1} dz\right\|
&\leq C\int_{\alpha}^{\infty}e^{-tx}(1+x)^{M-1}dx\nonumber\\
&\leq Ce^{-\alpha t}\left[\frac{1}{t}+\frac{1}{t^2}+
\cdots\frac{1}{t^{M}}\right]\nonumber\\
&\leq
\underbrace{Ce^{-\alpha t}\left[1+\frac{1}{t}+\frac{1}{t^2}+
\cdots\frac{1}{t^{M}}\right]}_{B(t)}.
\end{align}
Combining \eqref{esti1} and \eqref{esti2} we conclude that 
the Dunford integral \eqref{dunford_power} is bounded 
by $B(t)$. Since $\K$ has a compact resolvent, if 
there is any zero eigenvallue then it must have finite 
algebraic multiplicity (Theorem 6.29, Page 187 in 
\cite{kato2013perturbation}). This implies that the 
projection operator $\pi_0$ is a finite rank operator that 
admits the canonical form (in $L^2$)
\begin{equation}
\pi_0=\sum_{i=1}^n\alpha_i\langle \cdot,v_i\rangle u_i.
\end{equation}
On the other hand, since 
\begin{align}
\pi_0\K f=&\sum_{i=1}^n\alpha_i\langle \K f,v_i\rangle u_i\nonumber\\
=&\sum_{i=1}^n\alpha_i\langle f,\K^* v_i\rangle u_i, \nonumber
\end{align}
where $\K^*$ is the $L^2$-adjoint of $\K$, 
we have 
\begin{equation}
\|\pi_0\K\|\leq \sum_{i=1}^n|\alpha_i|\|\K^*v_i\|\|u_i\|.
\end{equation}
Hence, $\pi_0\K$ is a bounded operator. 
By using the Dunford integral representation 
over $\partial\S_{\K}$ it is straightforward to show that 
$e^{-t\K}\K$ is also a bounded operator for $t>0$. 
Combining these results with the triangle inequality, 
we have that for any fixed $t>0$ and any $n\in\mathbb{N}$
\begin{align}
\left|\left\|e^{-t\K/n}\K  \right\|-
\left\|\pi_0\K  \right\|\right|
\leq \left\|e^{-t\K/n}\K -\pi_0\K 
\right\|\leq B\left(\frac{t}{n}\right). 
\end{align} 
Finally, by using the operator identity $e^{-t\K}\K^n=(e^{-t\K/n}\K)^n$ we obtain
\begin{align}
\left\|e^{-t\K}\K^n \right\|\leq
\left\|e^{-t\K/n}\K \right\|^n\leq 
\left(B\left(\frac{t}{n}\right)+\left\|\pi_0\K \right\|\right)^n 
\qquad t>0,
\end{align} 
which completes the proof. 

\end{proof}

\noindent
The inequality \eqref{power_K} suggests the flow defined by 
the semigroup $e^{-t\K}$ has bounded 
derivatives in time. We emphasize that the estimate 
\eqref{power_K} is not sufficent to prove the 
convergence of the formal power series expansion of 
$e^{-t\K}$ since 
\begin{equation}
\lim_{n\rightarrow\infty}\frac{\|e^{-t\K}\K^n\|}{n!}\neq 0.
\end{equation} 
 
\subsection{Analysis of the projected Kolmogorov operator}
\label{sec:QKQ}

In this section we analyze the semigroup 
$e^{-t\Q\K\Q}$ generated by the operator $\Q\K\Q$, 
where $\K$ is the Kolmogorov operator \eqref{KI},  
$\P$ and $\Q=\I-\P$ are projection operators in $L^2(\R^n)$. Such semigroup appears 
in the EMZ memory and fluctuation terms  
(see Eqs. \eqref{gle_full}, \eqref{SFD} and \eqref{f}). 
In principle, the projection operator $\P$ and therefore the 
complementary projection $\Q$ can be chosen arbitrarily 
\cite{zhu2018estimation,chorin2002optimal}. 
Here we restrict our analysis to finite-rank symmetric
projections in $L^2(\R^n)$. Mori's projection \eqref{Mori_P} 
is one of such projections. 

\begin{theorem}\label{QKQ_cusp}
Let  $\P: L^2(\R^n)\rightarrow L^2(\R^n)$ be a 
finite-rank, symmetric projection operator. If $\K$ satisfies all
conditions listed in Theorem \ref{EH_spec_est}, then the 
operator $\Q\K\Q$ is also maximal accretive and has a compact 
resolvent. Moreover, the spectrum of $\Q\K\Q$ lies within 
the cusp
\begin{align}
\S_{\Q\K\Q}=\{z\in \mathbb{C}| \Re z\geq 0,|z+1|<(8C_{\Q})^{M_{\Q}/2}(1+\Re z)^{M_{\Q}}\}
\label{cusp1}
\end{align}
for some the positive constants $C_{\Q}$ 
and integer $M_{\Q}$.
\end{theorem}

\begin{proof}
We first show that if $\K$ is closely defined and 
maximal accretive, and that $\Q\K\Q$ has the same 
properties. According to the Lumer-Phillips theorem 
\cite{engel1999one}, the adjoint of a maximal-accretive 
operator is accretive, and therefore
\begin{align*}
\Re\langle\K f,f\rangle&\geq 0 \qquad \forall f\in D(\K),\\
\Re\langle\K^*f,f\rangle&\geq 0 \qquad \forall f\in D(\K^*).
\end{align*}
Here  $D(\K)$ and  $D(\K^*)$ denote the domain
of the linear operators $\K$ and $\K^*$, respectively 
(see, e.g., \cite{VenturiSpectral}). 
On the other hand, if $\P$ is a symmetric operator in 
$L^2(\R^n)$ then $\Q=\I-\P$ is also symmetric. This 
implies that
\begin{align*}
\Re\langle\Q\K\Q f,f\rangle&=\Re\langle\K\Q f,\Q f\rangle\geq 0
\qquad \forall f\in D(\K),\\
\Re\langle(\Q\K\Q)^* f,f\rangle&=
\Re\langle\K^*\Q f,\Q f\rangle\geq 0 \qquad \forall f\in D(\K^*),
\end{align*}
i.e., $\Q\K\Q$ and its adjoint $\Q\K^*\Q$ are 
both maximal-accretive. $\Q\K\Q$ is also a closable operator 
defined in $D(\K)$. This can be seen by decomposing it 
as $\Q\K\Q=\K-\K\P-\P\K\Q$. In fact, if 
$\K$ is a closed operator then $\Q\K\Q$ is 
closed since $\K\P$ and $\P\K\Q$ are bounded 
\cite{kato2013perturbation}, as we shall see hereafter.
By using the Lumer-Philips theorem, we conclude 
that $\Q\K\Q$ is also maximal accretive, and its closure 
generates a contraction semigroup $e^{-t\Q\K\Q}$ 
in $L^2(\R^n)$.
Next, we show that if $\K$ satisfies the hypoelliptic 
estimate $\| u\|_{\delta,\delta}\leq 
C(\| u\|+\|\K u\|)$, 
then so does $\Q\K\Q$, i.e.,
\begin{align}\label{QLQ_est4}
\|u\|_{\delta,\delta}\leq C(\|u\|+\|\Q\K\Q  u\|).
\end{align}
By using triangle inequality we obtain
\begin{align*}
\| u\|_{\delta,\delta}
\leq C(\|u\|+\|\K u\|)
\leq C(\|u\|+\|\K\P  u\|+\|\Q\K\Q u\|+\|\P\K\Q  u\|).
\end{align*}
To prove \eqref{QLQ_est4}, it is sufficient to show 
that $\K\P$ and $\P\K\Q$ are bounded operators 
in $L^2(\R^n)$. To this end, we recall that any 
finite-rank projection admits the canonical 
representation
\begin{equation}
\P=\sum_{i=1}^m\lambda_i\langle \cdot,\phi_i\rangle\varphi_i,
\end{equation} 
where $\{\phi_i\}_{i=1}^m$ and $\{\varphi_i\}_{i=1}^m$ 
are elements $L^2(\R^n)$. This implies that
\begin{align}
\|\K\P  u\|
&=\left\|\sum_{i=1}^m\lambda_i\langle u,
\phi_i\rangle\K\varphi_i\right\|\leq \sum_{i=1}^m |\lambda_i| 
\|\K\varphi_i\|\|\phi_i\|\| u\|= C\| u\|,\\
\|\P\K\Q  u\|
&=\left\|\sum_{i=1}^m\lambda_i\langle\K\Q  u,
\phi_i\rangle\varphi_i\right\|
=\left\|\sum_{i=1}^m\lambda_i\langle  u,\Q\K^*\phi_i\rangle\varphi_i\right\|
\leq \sum_{i=1}^m|\lambda_i|\|\Q\K^*\phi_i
\|\|\varphi_i\|\|u\|=C\| u\|.
\end{align}
This proves that $\K\P$ and $\P\K\Q$ are both 
bounded linear operators. 
At this point we notice that if $\Q\K\Q$ is accretive, 
then $(\Q\K\Q+\I)$ invertible. Moreover, since $\K$ 
is accretive we have
\begin{align*}
\|(\Q\K\Q+\I) u\|^2=\|\Q\K\Q\|^2+2\Re\langle\K\Q u,\Q u\rangle+\|u\|^2\geq \|\Q\K\Q\|^2+\|u\|^2.
\end{align*}
This implies that 
\begin{align*}
\|u\|_{\delta,\delta}\leq C(\| u\|+\|\Q\K\Q  u\|)\leq 
\sqrt{2}C\|(\Q\K\Q+\I) u\|\Rightarrow 
\|(\Q\K\Q+\I)^{-1}u\|_{\delta,\delta}\leq 
\sqrt{2}C\| u\|, 
\end{align*}
i.e., $(\Q\K\Q+\I)^{-1}$ is a bounded operator from 
$L^2$ into the weighted Sobolev space 
$\S^{\delta,\delta}$ defined in \eqref{wss}. 
At this point we recall that $S^{\delta,\delta}$ is compactly 
embedded into $L^2$ (Lemma 3.2 \cite{eckmann2003spectral}).
Hence, $(\Q\K\Q+\I)^{-1}$ is compact from $L^2$ into $L^2$ 
and therefore $\Q\K\Q$ has a compact resolvent 
\cite{kato2013perturbation}. 
To prove that the discrete spectrum of $\Q\K\Q$ 
lies within the cusp $\S_{\Q\K\Q}$ defined in \eqref{cusp1}, 
we follow the procedure outlined in \cite{eckmann2003spectral}. 
To this end, let $\K\in\Pol_2^N$. Then, for 
$\delta=\max\{2,N\}$ we have the bound
\begin{align*}
\|(\K+\I) u\|\leq C\|u\|_{\delta,\delta} \qquad \forall u\in
\mathscr{S}_n,
\end{align*}
and 
\begin{align*}
\|(\Q\K\Q+\I) u\|\leq\|\Q\|(\|\K u\|+\|\K\P u\|)+
\| u\|\leq C(\|\K u\|+\|u\|)\leq 
\sqrt{2}C\|(\K+\I) u\|\leq C\| u\|_{\delta,\delta}.
\end{align*}
Recall that $\Q\K\Q: D(\Q\K\Q)\rightarrow L^2(\R^{2d})$ 
is maximally accretive. Therefore, by Lemma 4.5 
in \cite{eckmann2003spectral}, for all $\delta>0$ we 
can find an integer $M_{\Q}>0$ and a constant $C$ 
such that 
\begin{align}\label{QLQ_est2}
\langle  u, [(\Q\K\Q+\I)^*(\Q\K\Q+\I)]^{1/M_{\Q}}\rangle
\leq C\| u\|^2_{\delta,\delta}.
\end{align}
By using the hypoelliptic estimate \eqref{QLQ_est2}, 
\eqref{QLQ_est4}, Proposition B.1 in \cite{nier2005hypoelliptic} 
and the triangle inequality we obtain
\begin{align*}
\frac{1}{4}|z+1|^{2/M_{\Q}}\|u\|^2&\leq 
C\|u\|^2_{\delta,\delta}+\|(\Q\K\Q-z)\|^2\\
&\leq C_{\Q}([1+\Re z]^2\|u\|^2+\|(\Q\K\Q-z) u\|^2).
\end{align*}
This result, together with 
the compactness of the resolvent of $\Q\K\Q$, implies 
that if $z\in\sigma(\Q\K\Q)$ (spectrum of $\Q\K\Q$) 
then 
\begin{align*}
\frac{1}{8}|z+
1|^{2/M_{\Q}}\| u\|^2<\frac{1}{4}|z+
1|^{2/M_{\Q}}\| u\|^2\leq C_{\Q}(1+\Re z)^2\|u\|^2.
\end{align*}
This proves that the spectrum of $\Q\K\Q$ is contained in the 
cusp-shaped region $\S_{\Q\K\Q}$ defined in equation \eqref{cusp1}. 
If $z\not\in\S_{\Q\K\Q}$, then we have resolvent estimate
\begin{align}\label{QKQ_reso}
\|(\Q\K\Q-z)^{-1}\|\leq \sqrt{8C_{\Q}}|z+1|^{-1/M_{\Q}}.
\end{align}

\end{proof}

\paragraph*{Remark} 
The main assumption at the basis of 
Theorem \ref{QKQ_cusp} is that $\P$ is a finite-rank 
symmetric projection.  Mori's projection \eqref{Mori_P} 
is one of such projections. If $\P$ is of finite-rank 
then both $\K\P$ and $\P\K\Q$ are bounded operators, 
which yields the hypoelliptic estimate \eqref{1epsilon_con}.
On the other hand, if $\P$ is an infinite-rank projection, 
e.g., Chorin's projection  
\cite{zhu2018estimation,chorin2000optimal,Chorin1,
zwanzig1973nonlinear}, then $\K\P$ and $\P\K\Q$ 
may not be bounded. Whether Theorem \ref{QKQ_cusp} 
holds for infinite-rank projections is an open question. 
\\

\noindent
With the resolvent estimate \eqref{QKQ_reso} 
available, we can now prove the analog of 
Theorem \ref{exp_est_K} and Corollary 
\ref{Taylor_K_conv}, with $\K$ replaced 
by $\Q\K\Q$. These results establish exponential 
relaxation to equilibrium of $e^{-t\Q\K\Q}$ and 
the regularity of the EMZ orthogonal dynamics 
induced by $e^{-t\Q\K\Q}$.

\begin{theorem}\label{exp_est_QKQ}
Assume that $\K$ satisfies all conditions listed in 
Theorem \ref{EH_spec_est}. Let 
$\P: L^2(\R^n)\rightarrow L^2(\R^n)$ 
be a symmetric finite-rank projection operator. 
If the spectrum of $\Q\K\Q$ in $L^2(\R^n)$ 
satisfies
\begin{align}\label{QKQ_spec_0}
\sigma(\Q\K\Q)\cap i\R\subseteq\{0\},
\end{align} 
then for any $0<\alpha_{\Q}<
\min (\Re \sigma(\Q\K\Q)/\{0\})$ there exits 
a positive constant $C=C(\alpha_{\Q})$ such that
\begin{align}\label{QKQ_estimation}
\|e^{-t\Q\K\Q}u_0-\pi^{\Q}_0u_0\|\leq 
Ce^{-\alpha_{\Q} t}\| u_0\|,
\end{align}
for all $ u_0\in L^2(\mathbb{R}^n)$ and for all $t>0$, 
where $\pi^{\Q}_0$ is the spectral projection 
onto the kernel of $\Q\K\Q$.
\end{theorem}

\begin{corollary} 
\label{powerC}
Suppose that $\P$ and $\K$ satisfy all 
conditions listed in Theorem \ref{exp_est_QKQ}. 
Then for any $t> 0$ the $n$-th order derivative 
of the semigroup $e^{-t\Q\K\Q}$ satisfies
\begin{align}\label{power_QKQ}
\|e^{-t\Q\K\Q}(\Q\K\Q)^nu_0\|\leq\left(\|\pi_0(\Q\K\Q)\|+B_{\Q}\left(\frac{t}{n}\right)\right)^n\|u_0\|,
\end{align}
where the function $B_{\Q}(t)$ has the same form as  
\eqref{Bt}, with $\alpha$ replaced by $\alpha_{\Q}$ 
and $M$ replaced by $M_\Q$.
\end{corollary}

The proofs of Theorem \ref{exp_est_QKQ} and Corollary \ref{powerC} closely follow the proofs of 
Theorem \ref{exp_est_K} and Corollary
\ref{Taylor_K_conv}. Therefore we omit them.
The semigroup estimate \eqref{QKQ_estimation} allows 
us to prove exponential convergence to the 
equilibrium state of the EMZ memory kernel and 
fluctuation force. Specifically, we have the following:

\begin{corollary} \label{cor:K_conver}
Consider a scalar observable  $u(t)=u(\bm x(t))$ 
with initial condition $u(0)=u_0$, and let $\P(\cdot)=\langle(\cdot),u_0\rangle u_0$ be a one-dimensional 
Mori's projection \eqref{Mori_P}. Then the EMZ memory 
kernel \eqref{SFD} converges exponentially fast 
to the equilibrium state $\langle\Q\K^*u_0,\pi^{\Q}_0\K u_0\rangle$, with rate $\alpha_{\Q}$. 
In other words, there exists 
a positive constant $C$ such that 
\begin{align}\label{kernel_est}
| K(t)-\langle\Q\K^*u_0,\pi^{\Q}_0\K u_0\rangle|\leq 
C e^{-\alpha_{\Q} t}.
\end{align}
\end{corollary}

\begin{proof}
A substitution of \eqref{QKQ_estimation} into \eqref{SFD} and 
subsequent application of Cauchy-Schwartz inequality
yields 
\begin{align}
|K(t)-\langle\Q\K^*u_0,\pi^{\Q}_0\K u_0\rangle|&
=|\langle u_0,\K e^{t\Q\K\Q}\Q\K u_0\rangle-\langle\Q\K^*u_0,\pi^{\Q}_0\K u_0\rangle|\nonumber\\
&=|\langle \Q\K^*u_0,e^{t\Q\K\Q}\K u_0\rangle-\langle\Q\K^*u_0,\pi^{\Q}_0\K u_0\rangle|\nonumber\\
&\leq C\|\Q\K^*u_0\|\|\K u_0\|e^{-\alpha_{\Q}t}.
\label{form_gen_K(t)}
\end{align}

\end{proof}

\noindent
It is straightforward to 
generalize Corollary \ref{cor:K_conver} to matrix-valued 
memory kernels \eqref{SFD} and obtain the following 
exponential convergence result
\begin{align}\label{matrixK}
\|\bm K(t)-\bm G^{-1}\bm C^{\Q}\|_{\mathscr{M}}\leq C\|\bm G^{-1}\bm D^{\Q}\|_{\mathscr{M}}e^{-\alpha_{\Q}t},
\end{align}
where $\|\cdot\|_{\mathscr{M}}$ denotes any matrix norm and  
$\bm G$ is the Gram matrix \eqref{gram}. Also, 
the matrix $\bm C^{\Q}$ has entries 
$C^{\Q}_{ij}=\langle\Q\K^*u_i(0),\pi^{\Q}_0\K u_j(0)\rangle$, 
while $D^{\Q}_{ij}=\|\Q\K u_i(0)\|\|\K u_j(0)\|$.  The proof of \eqref{matrixK} follows immediately 
from the following inequality
\begin{align}\label{est22}
\langle u_i(0),\K e^{t\Q\K\Q}\Q\K u_j(0)\rangle-\langle\Q\K^*u_i(0),\pi^{\Q}_0\K u_j(0)\rangle
&=\langle \Q\K^*u_i(0), e^{t\Q\K\Q}\K u_j(0)-
\pi^{\Q}_0\K u_j(0)\rangle\nonumber\\
&\leq C \|\Q\K u_i(0)\|\|\K u_j(0)\|e^{-\alpha_{\Q}t}.
\end{align}
In fact, a substitution of  \eqref{est22} into \eqref{SFD} 
yields \eqref{matrixK}. Similarly, we can prove that 
the fluctuation term \eqref{f} reaches the equilibrium state
exponentially fast in time. If we choose the 
initial condition as $\bm u_0=\Q\K \bm u_{0}$ then 
for all $j=1,...,m$, we have
\begin{align}\label{QKQ_estimation_force}
\|f_j(t)-\pi^{\Q}_0\Q\K u_{j}(0)\|=\|e^{-t\Q\K\Q}\Q\K u_{j}(0)-
\pi^{\Q}_0\Q\K u_{j}(0)\|\leq Ce^{-\alpha_{\Q} t}\|\Q\K u_{j}(0)\|.
\end{align}
Let us now introduce the tensor product space 
$\displaystyle V=\otimes_{i=1}^mL^2(\R^n)$ 
and the following norm 
\begin{align}\label{prod_norm}
\|\bm r(t)\|_{V}:=\left\|\left(\|r_1(t)\|,\|r_2(t)\|,\cdots,\|r_m(t)\|\right)\right\|_{\mathscr{M}},
\end{align}
where $\|\cdot\|$ is the standard $L^2(\R^n)$ norm, 
and $\|\cdot\|_\mathscr{M}$ is any matrix norm.  
Then from \eqref{QKQ_estimation_force} it 
follows that
\begin{align}
\| \bm f(t)-\pi^{\Q}_0\Q\K \bm u_0\|_{V}\leq Ce^{-\alpha_{\Q}t}\|\Q\K \bm u(0)\|_V.
\end{align}

\section{An application to Langevin dynamics}
\label{sec:app}

All results we obtained so far can be applied to 
stochastic differential equations of the 
form \eqref{eqn:sde}, provided the MZ 
projection operator is of finite-rank. 
In this section, we study in detail the 
Langevin dynamics of an interacting 
particle system widely used in statistical 
mechanics to model liquids and gasses \cite{lei2016data,Risken}, 
and show that the EMZ memory kernel \eqref{SFD} and 
fluctuation term \eqref{f} decay exponentially 
fast in time to a unique equilibrium state. 
Such state is defined by the projector operator 
$\pi_0^{\Q}$ appearing 
in Theorem \ref{exp_est_QKQ} and 
Corollary \ref{cor:K_conver}. 
Hereafter we will determine the exact 
expression of such projector for a system of 
interacting identical particles modeled 
by the following SDE in $\R^{2d}$ 
\begin{align}\label{eqn:LE}
\begin{dcases}
\frac{d\bm q}{dt }=\frac{1}{\mu}\bm p,\\
\frac{d\bm p}{dt}=-\nabla V(\bm q)-\frac{\gamma}{\mu}\bm p+\sigma \bm \xi(t),
\end{dcases}
\end{align}
where $\mu$ is the mass of each particle, 
$V(\bm q)$ is the interaction potential and 
$\bm \xi (t)$ is a $d$-dimensional Gaussian white 
noise process modeling physical Brownian motion.    
The parameters $\sigma$ and $\gamma$ represent, 
respectively, the amplitude of the fluctuations and 
the viscous dissipation coefficient. Such parameters  
are linked by the fluctuation-dissipation relation 
$\sigma=(2\gamma/\beta)^{1/2}$, where $\beta$ 
is proportional to the inverse of the thermodynamic 
temperature. The stochastic dynamical system 
\eqref{eqn:LE} is widely used in statistical mechanics 
to model the mesoscopic dynamics of liquids and gases.
Letting the mass $\mu$ in \eqref{eqn:LE} 
go to zero, and setting $\gamma=1$ yields 
the so-called overdamped Langevin dynamics, i.e., 
Langevin dynamics where no average acceleration 
takes place.
The (negative) Kolmogorov operator \eqref{KI} associated with 
the SDE \eqref{eqn:LE} is given by
\begin{align}\label{L:LE}
\K=-\frac{\bm p}{\mu}\cdot\nabla_{\bm q}+
\nabla_{\bm q}V(\bm q)\cdot\nabla_{\bm p}+
\gamma\left(\frac{\bm p}{\mu}\cdot\nabla_p-\frac{1}{\beta}\Delta_p\right),
\end{align}
where ``$\cdot$'' denotes the standard dot product. 
If the interaction potential $V(\bm q)$ is strictly positive 
at infinity then the Langevin equation \eqref{eqn:LE} 
admits an unique invariant Gibbs measure 
given by
\begin{equation}
\rho_{eq}(\bm p,\bm q)=\frac{1}{Z} e^{-\beta H(\bm p,\bm q)},
\end{equation}
where 
\begin{equation}
H(\bm p,\bm q)=\frac{\|\bm p\|_2^2}{2\mu}+V(\bm q), 
\end{equation}
is the Hamiltonian and $Z$ is the partition function. 
%
%
At this point we introduce the unitary transformation 
$\U:L^2(\R^{2d}) \rightarrow L^2(\R^{2d},\rho_{eq}) $ 
defined by 
\begin{align}\label{unitary_t}
(\U g)(\bm p,\bm q)=\sqrt{Z}e^{\beta H(\bm p,\bm q)/2}g(\bm p,\bm q),
\end{align}
where $L^2(\R^{2d};\rho_{eq})$ is a weighted Hilbert space 
endowed with the inner product 
\begin{equation}
\langle h,g\rangle_{\rho_{eq}}=\int h(\bm p,\bm q)  g(\bm p,\bm q) 
\rho_{eq}(\bm p,\bm q) d\bm pd\bm q.
\end{equation}
The linear transformation \eqref{unitary_t} is an 
isometric isomorphism  between 
the spaces $L^2(\R^{2d})$ and $L^2(\R^{2d};\rho_{eq})$.
In fact, for any $\tilde{u}\in L^2(\R^{2d})$, there exists 
a unique $u\in L^2(\R^{2d};\rho_{eq})$ such 
that $\tilde{u}=(e^{-\beta H/2}/\sqrt{Z})u$ 
and 
\begin{align}\label{iso}
\|\tilde u\|_{L^2}=\|u\|_{L^{2}_{eq}}.
\end{align}
By applying \eqref{unitary_t} to \eqref{L:LE} we 
construct the transformed Kolmogorov operator 
$\tK=\U^{-1}\K\U$, which has the explicit expression
\begin{align}\label{tK}
\tK=-\frac{\bm p}{\mu}\cdot\nabla_{\bm q}+\nabla V(\bm q)\cdot
\nabla_{\bm p}+\frac{\gamma}{\beta}\left(-\nabla_{\bm p}+
\frac{\beta}{2\mu}\bm p\right)\cdot\left(\nabla_{\bm p}+
\frac{\beta}{2\mu}\bm p\right).
\end{align}
This operator can be written in the canonical 
form \eqref{general_K} as
\begin{equation}
\tK= \sum_{i=1}^d \X_i^*\X_i-\X_0,
\label{tKseries}
\end{equation}
provided we set  
\begin{equation}
\begin{dcases}
\X_0=\frac{\bm p}{\mu}\cdot\nabla_{\bm q}-\nabla V(\bm q)\cdot\nabla_{\bm p},\\
\X_i=\sqrt{\frac{\gamma}{\beta}}\left(\partial_{p_i}+\frac{\beta}{2\mu}p_i\right), \\
\X_i^*=\sqrt{\frac{\gamma}{\beta}}\left(-\partial_{p_i}+\frac{\beta}{2\mu}p_i\right).
\end{dcases}
\label{X_i_Langevin}
\end{equation}
Note that $\X_0$ is skew-symmetric in 
$L^2(\R^{2d})$. Also, $\X_i^*$ and $\X_i$ 
can be interpreted as creation and annihilation 
operators, similarly to a harmonic quantum 
oscillator \cite{Justin}. 
The Kolmogorov operator $\tK$ and its formal adjoint $\tK^*$ 
are both accretive, closable and with maximally accretive 
closure in  $L^2(\R^{2d})$ (see, e.g.,  \cite{herau2004isotropic,nier2005hypoelliptic,
eckmann2000non})
Similar to the Kolmogorov operator $\tK=\U^{-1}\K\U$, 
we can transform the MZ projection operators $\P$ 
and $\Q$ into operators in the ``flat'' Hilbert space 
$L^2(\R^{2d})$ as $\tilde{\P}=\U^{-1}\P\U$ and 
$\tilde{\Q}=\U^{-1}\Q\U$. The relationship 
between $L^2(\R^{2d})$, $L^{2}(\R^{2d};\rho_{eq})$ 
and the operators defined between such spaces can 
be summarized by the following commutative diagram
\begin{equation}
\begin{tikzpicture}
  \matrix (m) [matrix of math nodes,row sep=3em,column sep=4em,minimum width=2em]
  {
     L^2(\R^{2d}) & L^2(\R^{2d};\rho_{eq}) \\
     L^2(\R^{2d}) & L^2(\R^{2d};\rho_{eq}) \\};
  \path[-stealth]
    (m-1-1) edge node [left] {$\tilde{\P},\tilde{\K},\tilde{\Q}$} (m-2-1)
            edge node [above] {$\U$} (m-1-2)        
    (m-2-2.west|-m-2-1) edge node [below] {$\U^{-1}$}
             (m-2-1)
    (m-1-2) edge node [right] {$\P,\K,\Q$} (m-2-2)
            ;
\end{tikzpicture}
\nonumber
\end{equation}
The properties of all operators in $L^2(\R^{2d})$ 
and $L^2(\R^{2d};\rho_{eq})$ are essentially the same 
since $\U$ is a bijective isometry. For instance if $\P$ is 
compact and symmetric then $\tP$ is also a compact 
and symmetric operator.

Next, we apply the analytical results we obtained in 
Section \ref{sec:K} and Section \ref{sec:QKQ} to the 
particle system described by the SDE \eqref{eqn:LE}. 
To this end, we just need to verify whether
$\tK$ is a poly-H\"ormander operator, i.e., if 
the operators  $\{\X_i\}_{i=0}^{d}$ appearing in 
\eqref{tKseries}-\eqref{X_i_Langevin} satisfy 
the poly-H\"ormander conditions in 
Proposition \ref{Lie_cond} and the estimate in
Theorem \ref{EH_spec_est} (see Section \ref{sec:K}).
This can be achieved by imposing additional 
conditions on the particle interaction potential $V(\bm q)$
(see \cite[Proposition 3.7]{eckmann2000non}).  
In particular, following Helffer and Nier \cite{nier2005hypoelliptic}, 
we assume that $V(\bm q)$ satisfies the following weak ellipticity 
hypothesis

\begin{hyp}\label{Hypo:V(q)}
The particle interaction potential $V(\bm q)$ 
is of class $C^{\infty}(\R^{d})$, and 
for all $\bm q\in \R^{d}$ it satisfies the following conditions:
\begin{enumerate}
\item $\forall \bm \alpha\in\mathbb{N}^d$ such that $|\bm \alpha|=1$, 
$|\partial_{\bm q}^{\bm \alpha}
V(\bm q)|\leq C_{\bm \alpha}\sqrt{1+\|\nabla V(\bm q)\|^2}$ 
for some positive constant $C_{\bm \alpha}$,
\item There exists $M\in \mathbb{N}$, and $C \geq 1$, such that 
$C^{-1}(1+\|\bm q\|^2)^{1/(2M)}\leq \sqrt{1+\|\nabla V(\bm q)\|^2}\leq 
C(1+\|\bm q\|^2)^{M/2}$.
\end{enumerate}
\end{hyp}

\noindent
Hypothesis \ref{Hypo:V(q)} holds for any particle 
interaction potential that grows at most polynomially at infinity, i.e., 
$V(\bm q)\simeq\|\bm q\|^{M}$ as  $\left\|\bm q\right\|\rightarrow\infty$. 
With this hypothesis, it is possible to prove the following

\begin{prop}[Helffer and Nier \cite{nier2005hypoelliptic}]
\label{prop1}
Consider the Langevin equation \eqref{eqn:LE} with particle 
interaction potential $V(\bm q)$ satisfying Hypothesis 
\ref{Hypo:V(q)}. Then the operator $\tK$ defined in \eqref{tK} has 
a compact resolvent, and a discrete spectrum bounded by the cusp 
$\S_{\K}$. Moreover, there exists a positive constant 
$C$ such that the estimate
\begin{align}\label{K_estimation11}
\|e^{-t\tK}u_0-\tilde{\pi}_0\tilde{u}_0\|\leq Ce^{-\alpha t}\|\tilde{u}_0\|
\end{align}
holds for all $\tilde{ u}_0\in L^2(\R^{2d})$ and for 
all $t>0$, where $\tilde{\pi}_0$ is the orthogonal projection 
onto the kernel of $\tK$ in $L^2(\R^{2d})$.
\end{prop}

\noindent
By using the isomorphism \eqref{unitary_t} we can  
rewrite Proposition \ref{prop1} in $L^2(\R^{2d};\rho_{eq})$ 
as
\begin{align}\label{iso_E}
\|e^{-t\K} u_0-\pi_0u_0\|_{L^2_{eq}}=
\|e^{-t\tilde{\K}}\tilde{u}_0-\tilde{\pi}_0\tilde{u}_0\|_{L^2}\leq Ce^{-\alpha t}\|\tilde{u}_0\|_{L^2}=Ce^{-\alpha t}\|u_0\|_{L^2_{eq}},
\end{align}
where $\pi_0=\U\tilde{\pi}_0\U^{-1}$ is the orthogonal 
projection  $\pi_0(\cdot)=\mathbb{E}[(\cdot)]$.
The inequality \eqref{iso_E} is completely equivalent 
to the estimate \eqref{K_estimation}.
It is also possible to obtain a prior estimate on the 
convergence rate $\alpha$ by building a
connection between the Kolmogorov operator and the Witten Laplacian 
(see \cite{herau2004isotropic,nier2005hypoelliptic} for further 
details).

Our next task is to derive an estimate for the operator 
$\tQ\tK\tQ$, and for the semigroup $e^{-t\tQ\tK\tQ}$ 
generated by the closure of  $\tQ\tK\tQ$.
According to Theorem \ref{QKQ_cusp} the 
spectrum of $\tQ\tK\tQ$ is bounded by the cusp 
$\S_{\tQ\tK\tQ}$, provided that $\P$ is an orthogonal 
finite-rank projection operator. 
On the other hand, Theorem \ref{exp_est_QKQ}
establishes exponential convergence of 
$e^{-t\tQ\tK\tQ}$ to equilibrium if $\tQ\tK\tQ$ 
satisfies condition \eqref{QKQ_spec_0}. 
It is left to
determine the exact form of the spectral 
projection $\tilde{\pi}_0^{\tQ}$, i.e., the projection 
onto the kernel of $\tQ\tK\tQ$ (see 
Theorem \ref{exp_est_QKQ}) and verify condition \eqref{QKQ_spec_0}. To this end, we consider 
a general Mori-type projection $\P$ and its unitarily 
equivalent version $\tP=\U^{-1}\P\U$
\begin{align}\label{projection_fq}
\P(\cdot)=\sum_{i=1}^m\langle \cdot,v_i\rangle_{\rho_{eq}} 
v_i,\qquad 
\tilde{\P}(\cdot)=\sum_{i=1}^m\langle \cdot ,v_i\rangle_{\rho_{eq}/2} v_ie^{-\beta H/2},
\end{align}
where $\{v_j\}_{j=1}^m=\{v_j(\bm q,\bm p)\}_{j=1}^m$ are zero-mean, i.e. $\langle v_i\rangle_{\rho_{eq}} =0$, orthonormal basis functions.
In \eqref{projection_fq} we used the shorthand 
notation 
\begin{equation}
\langle h\rangle_{\rho_{eq}/2}=\frac{1}{Z}\int g(\bm p,\bm q)e^{-\beta H(\bm p,\bm q)/2}d\bm pd\bm q.
\end{equation}
\begin{lemma}\label{lemma1}
{Suppose that the particle interaction potential $V(\bm q)$ in 
\eqref{L:LE} satisfies Hypothesis \ref{Hypo:V(q)}. Then for any set of 
observables $\{w_j\}_{j=1}^m$ satisfying $\langle w_j, v_i\rangle_{\rho_{eq}} =0$ 
and $\K w_j=v_j$ we have that the kernel of $\tQ\tK\tQ$ 
is given by 
\begin{align}
\textrm{Ker}(\tQ\tK\tQ)&=\text{Ker}(\tK)\,\cup\, \text{Ran}(\tP)\,\cup\, \text{Span}\{w_je^{-\beta H/2}\}_{j=1}^m,\label{1}
\end{align}}
where $\tK$ and $\tP$ are defined in \eqref{tK} 
and \eqref{projection_fq}, respectively. In particular, if
$\P$ is defined as $\P(\cdot)=\langle \cdot,p_j\rangle_{\rho_{eq}} p_j$, 
where $p_j$ is the momentum of $j$-th particle, then 
we have 
\begin{align}
\sigma(\tQ\tK\tQ)\cap i\R&\subseteq\{0\}.\label{2}
\end{align}
\end{lemma}

\begin{proof}
We first prove \eqref{1}. {To this end, let us first define the finite-dimensional space
\begin{equation}
W=\text{Ker}(\tK)\,{\cup}\, \text{Ran}(\tP)\,{\cup}\, \text{Span}\{w_je^{-\beta H/2}\}_{j=1}^m.
\end{equation}
}
If $u\in \text{Ker}(\tQ\tK\tQ)$ then $\tQ\tK\tQ u=0$. This implies $ \tK\tQ u-\tP\tK\tQ u=0$. Equivalently, 
\begin{equation}
\begin{aligned}
\tK u&=\tK\tP u+\tP\tK\tQ u\\
&=\sum_{j=1}^m\langle u,v_i\rangle_{eq/2} \tK v_ie^{-\beta H/2}+\sum_{j=1}^m\langle\tK\tQ u,v_i\rangle_{eq/2} v_ie^{-\beta H/2}\\
& \in \text{Span}\{v_je^{-\beta H/2}\}_{j=1}^m\,{\cup}\, \text{Span}\{\tK v_je^{-\beta H/2}\}_{j=1}^m.
\end{aligned}
\end{equation}
Since $\K w_j=\U\tK\U^{-1}w_j=v_j$, we have $\tK w_je^{-\beta H/2}=v_je^{-\beta H/2}$. 
This implies that $u\in {W}$ and $\text{Ker}(\tQ\tK\tQ)\subseteq {W}$.  
Let $f$ be an arbitrary element in $W$.  Then, 
\begin{align*}
f=\alpha e^{-\beta H/2}+\sum_{j=1}^m\rho_jv_j
e^{-\beta H/2}+\sum_{j=1}^m\theta_jw_je^{-\beta H/2}.
\end{align*}
$\{\alpha, \rho_1,\ldots,\rho_m, 
\theta_1,\ldots,\theta_m\}$ 
are the coordinates of $f$ in the finite-dimensional 
space $W$. 
By using the definition of $\P$, the fact 
that $\langle v_i\rangle_{\rho_{eq}} =\langle v_i,w_j\rangle_{\rho_{eq}} =0$ 
and $\langle v_j^2 \rangle_{\rho_{eq}} =1$ we obtain 
\begin{align*}
\tP f=\sum_{j=1}^m\rho_jv_je^{-\beta H/2}
&\Rightarrow 
\tQ f=\alpha e^{-\beta H/2}+\sum_{j=1}^m\theta_jw_j
e^{-\beta H/2}.
\end{align*}
Therefore,
\begin{align*}
\tK \tQ f=\tP\tK \tQ f=\sum_{j=1}^m\theta_jv_je^{-\beta H/2}
&\Rightarrow \tQ\tK \tQ f=0.
\end{align*}
This proves that $W\subseteq \text{Ker}(\tQ\tK\tQ)$, 
{and therefore \eqref{1} holds. In fact,} 
the kernel of $\tQ\tK\tQ$ 
can be constructed by taking the union of 
three sets defined by the conditions:
\begin{enumerate} 
\item $\tQ u=0$, which implies $\tP u=0$, i.e., $u\in \text{Ran}(\tP)$;
\item $\tQ u\neq 0,\tK\tQ u=0$, which implies 
$\tQ u\in \text{Ker}(\tK)$. 
{This is possible only if $u\in \text{Ran}(\tQ)\cap\text{Ker}(\tK)$ 
since in this case we have $\tQ u=u$;}
\item $\tQ u\neq 0,\tK\tQ u\neq 0, \tQ\tK\tQ u=0$, 
which implies $\tK\tQ u=\tP\tK\tQ u\neq 0$. {This is 
possible only if $\tQ u=u$, $\tK u\neq 0$ and 
$u\in \text{Span}\{ w_ie^{-\beta H/2}\}_{i=1}^m$, provided 
that the set of observables $\{w_j\}_{j=1}^m$ 
satisfies $\langle w_j, v_i\rangle_{\rho_{eq}} =0$ and $\K w_j=v_j$.}
\end{enumerate}
{
Combining these three cases and using 
the fact that $L^2(\R^{2d})=\text{Ran}(\tP)\oplus\text{Ran}(\tQ)$ we have 
\begin{align*}
\textrm{Ker}(\tQ\tK\tQ)
&=\text{Ran}(\tP)\oplus\left(
\text{Ker}(\tK)\cap\text{Ran}(\tQ)\right)\oplus \left(\text{Ker}^{\perp}(\tK)\cap \text{Ran}(\tQ) \cap\text{Span}\{w_je^{-\beta H/2}\}_{j=1}^m\right),\\
&=\text{Ran}(\tP)\cup
\text{Ker}(\tK)\cup\text{Span}\{w_je^{-\beta H/2}\}_{j=1}^m.
\end{align*}
}
Next, we prove condition \eqref{2}  for 
$\P=\langle\cdot,p_i\rangle_{\rho_{eq}} p_j$.
Such condition states that the only eigenvalue 
of $\tQ\tK\tQ$ on the imaginary 
axis $i\R$ is the origin. Equivalently, this 
means that for all $u\in L^2(\R^{2d})$ such 
that $\tQ\tK\tQ u=i\lambda u$ ($\lambda\in\R$)
we have that $\lambda=0$. To see this, we first notice 
that $\Re(\tQ\tK\tQ)u=0$. Since $\tQ$ 
is a symmetric operator, we have that 
$\Re(\tQ\tK\tQ)u=[\tQ(\tK+\tK^*)\tQ]u/2=\tQ\tS\tQ u=0$, 
where $\tS=\sum_{j=1}^d\X_i^*\X_i$.  
This means that $u\in \text{Ker}(\tQ\tS\tQ)$.
As before, $\text{Ker}(\tQ\tS\tQ)$ can be constructed by 
taking the union of three different sets defined by the conditions: 
\begin{enumerate} 
\item $\tQ u=0$, which implies $u=\rho p_j$; 
\item $\tQ u\neq 0,\tS\tQ u=0$, which imply 
$u\in \text{Ker}(\S)$, i.e.,   $u=\alpha \Phi(\bm q)
e^{-\frac{\beta}{4\mu}\|\bm p\|^2}$, where 
$\Phi(\bm q)$ is an arbitrary function of the 
coordinates $\bm q$; 
\item $\tQ u\neq 0,\tS\tQ u\neq 0, \tQ\tS\tQ u=0$, 
which imply $\tP\tS\tQ u=\tS\tQ u$.
\end{enumerate}
The first condition implies that $\tQ\tK\tQ u=0=i\lambda u$, i.e., 
$\lambda=0$. Upon definition of $g=\tQ u$, the third condition 
implies that $\langle \tS g,p_j\rangle_{eq/2} p_j
e^{-\beta H/2}=\tS g$. 
This is  a linear ODE for $g$ that has the unique 
solution $g=\theta p_je^{-\beta H/2}$ for some 
constant $\theta\neq 0$. However, it is easy to show that 
there is no $u$ such that $\tQ u=g=\theta p_j
e^{-\beta H/2}$. In fact, if such $u$ exists 
then  $\tP\tQ u=\tP g=\theta p_je^{-\beta H/2}\neq 0$ 
which contradicts the operator identity $\tP\tQ=0$.  
Lastly, the second conditions implies that if 
$u=\Phi(\bm q)e^{-\frac{\beta}{4\mu}\|\bm p\|^2}$ then 
$\tP u=0$ and $\tQ u=u$. 
Now consider $\Im(\tQ\tK\tQ)u=\tQ\X_0\tQ u=i\lambda u$. 
By using the conditions above we obtain 
\begin{align}
\tQ\X_0\tQ u&=\X_0\tQ u-\tP\X_0\tQ u\nonumber\\
&=\X_0u-\langle \X_0\tQ u,p_j\rangle_{eq/2}p_je^{-\beta H/2}\nonumber\\
&=\sum_{i=1}^d-\frac{\beta}{2\mu}p_i\partial_{q_i}V(\bm q)\Phi(\bm q)e^{-\frac{\beta}{4\mu}\|\bm p\|^2}
-\frac{p_i}{\mu}\partial_{q_i}\Phi(\bm q)e^{-\frac{\beta}{4\mu}\|\bm p\|^2}-\langle\X_0\tQ u,p_j\rangle_{eq/2}p_je^{-\beta H/2}
\nonumber\\
&=\sum_{i=1}^dp_i(f_i(\bm q))e^{-\frac{\beta}{4\mu}\|\bm p\|^2}\nonumber\\
&=i\lambda \Phi(\bm q)e^{-\frac{\beta}{4\mu}\|\bm p\|^2}.
\end{align}
The last equality holds if and only if $f_i(\bm q)=0$ 
and $\lambda=0$. 
This proves that $\tQ\tK\tQ$ has no purely 
imaginary eigenvalues. 
 
\end{proof}

\paragraph{Remark}
Proving the existence and 
uniqueness of a set of observables $\{w_1,\ldots,w_m\}$ 
such that $\langle w_j, v_i\rangle_{\rho_{eq}} =0$ 
and $\K w_j=v_j$ is not straightforward as 
it involves the analysis of a system of $m$ 
hypo-elliptic equations $\K w_j=v_j$. Fortunately, 
this can avoided in some cases, e.g., when the observable 
$v_j$ coincides with time derivative of $w_j$. A 
typical example is the momentum $p_j$ of the 
$j$-th particle. We also emphasize that 
in Lemma \ref{lemma1} we proved that 
$\Q\K\Q$ has no purely imaginary eigenvalues 
if the projection operator $\P$ is chosen 
as $\P=\langle \cdot,p_j \rangle_{\rho_{eq}}p_j$. 
This result may not be true for other projections, i.e., 
$\Q\K\Q$ can, in general, have purely imaginary eigenvalues. \\

\noindent
Lemma \ref{lemma1} allows us to prove 
the following exponential convergence result for the 
semigroup $e^{-t\Q\K\Q}$.

\begin{prop}\label{prop2}
Suppose that the particle interaction potential $V(\bm q)$ in 
\eqref{L:LE} satisfies Hypothesis \ref{Hypo:V(q)}. Let $\P$ be
the projection operator \eqref{projection_fq}. {For any set of observables 
$\{w_1,\ldots,w_m\}$ satisfying $\langle w_j, v_i\rangle_{\rho_{eq}} =0$, 
$\K w_j=\K^{*}w_j=v_j$ and $\sigma(\Q\K\Q)\cap i\R
\subseteq\{0\}$} there exist two positive 
constants $C$ and $\alpha_{\Q}$ such that 
\begin{align}\label{QKQ_estimation1}
\|e^{-t\Q\K\Q}u_0-\pi^{\Q}_0u_0\|_{L^2_{eq}}\leq Ce^{-\alpha_{\Q} t}\| u_0\|_{L^2_{eq}}
\end{align}
for all $u_0\in L^2(\R^{2d};\rho_{eq})$ and $t>0$. In \eqref{QKQ_estimation1},
{$\pi^{\Q}_0$ is the orthogonal projection onto the linear space 
$\text{Ker}(\Q\K\Q)=\text{Ker}(\K)\,{\cup}\, \text{Ran}(\P)\, {\cup}\,  \text{Span}\{w_j\}_{j=1}^{m}$.}
\end{prop}

\begin{proof} Rewrite \eqref{QKQ_estimation1} 
as an $L^2(\R^{2d})$ estimation problem  
\begin{align}\label{equi_estimation}
\|e^{-t\tQ\tK\tQ}\tilde{u}_0-\tilde{\pi}^{\tQ}_0
\tilde{u}_0\|_{L^2}
\leq Ce^{-\alpha_{\Q} t}\|\tilde{ u}_0\|_{L^2},
\end{align}
where $\tilde{\pi}^{\tQ}_0=\U^{-1}\pi_0^{\Q}\U$. 
The transformed Kolmogorov operator $\tK$ is 
of the form \eqref{general_K} with 
compact resolvent and a spectrum enclosed 
in cusp-shaped region of the complex plane 
shown in Figure \ref{fig:cusp} (see Proposition \ref{prop1}). 
Then, by Theorem \ref{QKQ_cusp}, the 
operator $\tQ\tK\tQ$ has exactly the same 
properties, provided $\tP$ is a symmetric, finite-rank projection.
To derive the estimate \eqref{QKQ_estimation1} we
simply use the conclusions of 
Theorem \ref{exp_est_QKQ}. To this end, we need 
to make sure that the following two conditions are satisfied
\begin{enumerate}
\item[] {\em Condition 1.}  $\text{Ran}(\pi_0^{\Q})=\text{Ker}(\tQ\tK\tQ)=\text{Ker}(\tK)\, {\cup}\, \text{Ran}(\tP)\, {\cup}\, \text{Span}\{w_je^{-\beta H/2}\}_{j=1}^{m}$. Moreover, the $L^2$-orthogonal space $\text{Ker}(\tQ\tK\tQ)^{\perp}$ is an invariant subspace of operator $\tQ\tK\tQ$.
 
\item [] {\em Condition 2.} $\tilde{\pi}^{\tQ}_0$ is an orthogonal 
projection in $L^2(\R^{2d})$.
\end{enumerate}

\paragraph*{Proof of Condition 1} 
In Lemma \ref{lemma1}, we have shown that 
$\text{Ker}(\tQ\tK\tQ)=\text{Ker}(\tK)\, {\cup}\, \text{Ran}(\tP)\, {\cup}\, 
\text{Span}\{w_je^{-\beta H/2}\}_{j=1}^{m}$. 
Hence, we just need to prove that $\text{Ker}(\tQ\tK\tQ)^{\perp}$
is an invariant subspace of operator $\tQ\tK\tQ$. 
To this end, we recall that the projection 
operator $\tP$ is a symmetric operator, 
therefore $\text{Ran}(\tP)=\text{Ran}(\tP^*)$. In 
\cite{nier2005hypoelliptic}, Helffer and Nier 
proved that $\text{Ker}(\tK)=\text{Ker}(\tK^*)=e^{-\beta H/2}$. 
By following the same mathematical steps that lead us 
to equation \eqref{1} we obtain
\begin{align}\label{Ker_QKQ}
\text{Ker}(\tQ\tK\tQ)=\text{Ker}(\tQ\tK^*\tQ)=\text{Ker}(\tK)\, {\cup}\, \text{Ran}(\tP)\, {\cup}\, \text{Span}\{w_je^{-\beta H/2}\}_{j=1}^{m}.
\end{align} 
We now verify that $\tQ\tK\tQ$ maps the 
linear subspace $\text{Ker}(\tQ\tK\tQ)^{\perp}$ into itself, i.e., 
that for any $u\in \text{Ker}(\tQ\tK\tQ)^{\perp}$ we have that 
$\tQ\tK\tQ u\in \text{Ker}(\tQ\tK\Q)^{\perp}$.  
To this end, we notice that if 
$w\in \text{Ker}(\tQ\tK\tQ)$ then  
\begin{align*}
\langle \tQ\tK\tQ u,w\rangle=\langle  u,\tQ\tK^*\tQ w\rangle=0.
\end{align*}
This follows directly from $\text{Ker}(\tQ\tK\tQ)=\text{Ker}(\tQ\tK^*\tQ)$.
On the other hand, if 
$u\in \text{Ker}(\tQ\tK\tQ)^{\perp}$ then $\tQ\tK\tQ u\neq 0$ 
and therefore we must have 
$\tQ\tK\tQ u\in \text{Ker}(\tQ\tK\tQ)^{\perp}$. 
Next, consider the following orthogonal decomposition of the Hilbert space $L^2(\R^{2d})$
\begin{align*}
L^2(\R^{2d})=\text{Ker}(\tQ\tK\tQ)\overset{\perp}{\oplus} \text{Ker}(\tQ\tK\tQ)^{\perp}.
\end{align*}
If we define a projection operator $\pi_0^{\tQ}$ 
with range $\text{Ran}(\pi_0^{\tQ})=\text{Ker}(\tQ\tK\tQ)$, 
then for any $\tilde {u}_0\in L^2(\R^{2d})$, we have 
the orthogonal decomposition 
\begin{align*}
\tilde{u}_0=\pi_0^{\tQ}\tilde{u}_0+(\tilde{u}_0-\pi_0^{\tQ}\tilde{u}_0), \quad \text{where}\quad \pi_0^{\tQ}\tilde{u}_0\in \text{Ker}(\tQ\tK\tQ), \quad \tilde{u}_0-\pi_0^{\Q}
\tilde{u}_0\in \text{Ker}(\tQ\tK\tQ)^{\perp}.
\end{align*}
Since $\text{Ker}(\tQ\tK\tQ)^{\perp}$ is an invariant 
subspace of $\tQ\tK\tQ$, and therefore of 
$e^{-t\tQ\tK\tQ}$, we have that
$e^{-t\tQ\tK\tQ}(\tilde{u}_0-\pi_0^{\Q}
\tilde{u}_0)\in \text{Ker}(\tQ\tK\tQ)^{\perp}$ for all $t>0$. 
On the other hand, since $\U$ is an unitary transformation 
we have $\sigma(\tQ\tK\tQ)\cap i\R=\{0\}$. These facts 
allow us to deform the domain of the Dunford integral 
representing $e^{t\tQ\tK\tQ}\tilde{u}_0-\tilde{\pi}_0^{\tQ}
\tilde{u}_0$  from $[-i\infty,+i\infty]$ to the cusp 
$\S'_{\tQ\tK\tQ}$, as we did in 
Theorem \ref{exp_est_K}. This yields
\begin{align*}
e^{t\tQ\tK\tQ}\tilde{u}_0-\tilde{\pi}_0^{\tQ}\tilde{u}_0=
e^{t\tQ\tK\tQ}\left(\tilde{u}_0-\tilde{\pi}_0^{\tQ}\tilde{u}_0\right)=\frac{1}{2\pi i}\int_{\partial \S'_{\tQ\tK\tQ}}e^{-tz}
\left(z-\tQ\tK\tQ\right)^{-1}\tilde{u}_0dz.
\end{align*}
At this point we can follow the exact same procedure in 
the proofs of Theorem \ref{exp_est_K} and Theorem 
\ref{QKQ_cusp} to show that the semigroup estimate \eqref{QKQ_estimation1} holds true.

\paragraph*{Proof of Condition 2} We first recall that 
$\text{Ran}(\tilde{\pi}_0^{\tQ})=\text{Ker}(\tQ\tK\tQ)$. 
This implies that for all $u\in L^2$ and all
$w\in \text{Ker}(\tQ\tK\tQ)^{\perp}$ we have 
\begin{align*}
\langle\tilde{\pi}_0^{\tQ}u,w\rangle=\langle u,[\tilde{\pi}_0^{\tQ}]^*w\rangle=0.
\end{align*}
Hence, $[\tilde{\pi}_0^{\tQ}]^*w=0$ for 
all $w\in \text{Ker}([\tilde{\pi}_0^{\tQ}]^*)$, which implies 
that $\text{Ker}(\tQ\tK\tQ)^{\perp}\subset 
\text{Ker}([\tilde{\pi}_0^{\tQ}]^*)$. On the other hand, 
for all $u\in L^2$ and all $w\in \text{Ker}([\tilde{\pi}_0^{\tQ}]^*)$, 
we have
\begin{align*}
\langle u,[\tilde{\pi}_0^{\tQ}]^*w\rangle=\langle\tilde{\pi}_0^{\tQ}u,w\rangle=0.
\end{align*}
From this equations it follows that 
$\text{Ker}([\tilde{\pi}_0^{\tQ}]^*)=\text{Ker}(\tQ\tK\tQ)^{\perp}$. 
Next we decompose $L^2(\R^{2d})$ as 
\begin{align*}
L^2(\R^{2d})=\text{Ker}\left(\tilde{\pi}_0^{\tQ}\right)\oplus 
\text{Ran}\left(\tilde{\pi}_0^{\tQ}\right),\qquad
L^2(\R^{2d})=\text{Ker}\left(\left[\tilde{\pi}_0^{\tQ}\right]^*\right)\oplus 
\text{Ran}\left(\left[\tilde{\pi}_0^{\tQ}\right]^*\right)
\end{align*}
It follows from the above result that $\text{Ran}(\tilde{\pi}_0^{\tQ})=\text{Ran}([\tilde{\pi}_0^{\tQ}]^*)=\text{Ker}(\tQ\tK\tQ)$ 
and $\text{Ker}(\tilde{\pi}_0^{\tQ})=\text{Ker}([\tilde{\pi}_0^{\tQ}]^*)=\text{Ker}(\tQ\tK\tQ)^{\perp}$. 
For all $u,w\in L^2(\R^{2d})$ we have that  
$w-\tilde{\pi}_0^{\tQ}w\in \text{Ker}(\tilde{\pi}_0^{\tQ})$, which 
can be written as 
\begin{align*}
\langle \tilde{\pi}_0^{\tQ}u,w-\tilde{\pi}_0^{\tQ}w\rangle=
\langle u,[\tilde{\pi}_0^{\tQ}]^*(w-\tilde{\pi}_0^{\tQ}w)\rangle
=0.
\end{align*}
Therefore the operator $\tilde{\pi}_0^{\tQ}$ is an orthogonal projection. 
This completes the proof. {In addition, since $\tilde{\pi}_0^{\tQ}$ has 
range $\text{Ker}(\tK)\, {\cup}\, \text{Ran}(\tP)\, {\cup}\, \text{Span}\{w_je^{-\beta H/2}\}_{j=1}^{m}$ it can be shown that for the special case $\langle v_i e^{-\beta H}\rangle _{\rho_{eq}}=\langle w_i e^{-\beta H}\rangle _{\rho_{eq}}=0$ and $\langle v_i, w_j\rangle_{\rho_{eq}}=0$ we have that  
$\tilde{\pi}_0^{\tQ}$ admits the explicit representation
\begin{align}\label{PiQ*}
\tilde{\pi}^{\tQ}_0=\tilde{\pi}_0+\tP+\sum_{i=1}^m\langle \cdot,w_i\rangle_{\rho_{eq}}w_ie^{-\beta H/2}.
\end{align}
The projection $\tilde{\pi}^{\tQ}_0$ can be transformed back to ${\pi}^{\Q}_0$ by using the mapping $\mathcal{U}$ defined in \eqref{unitary_t}.}

\end{proof}

{\paragraph{Remark} In general, the orthogonal projection onto 
$\text{Ker}(\tK)\, {\cup}\, \text{Ran}(\tP)\,{\cup}\, \text{Span}\{w_je^{-\beta H/2}\}_{j=1}^{m}$ can be written as
\begin{align}\label{simple_form}
    \tilde{\pi}_0^{\tQ}=\sum_{i=1}^{2m+1}\langle\cdot, \tilde{e}_i\rangle \tilde{e}_i,
\end{align}
where $\{\tilde{e}_i\}_{i=1}^{2m+1}$ is an orthonormal basis of $\text{Ker}(\tQ\tK\tQ)$ in $L^2(\R^{2d})$. 
}

\paragraph{Remark}
In Proposition \ref{prop2}, we assumed 
that $\K w_j=\K^* w_j=v_j$. If this condition is 
not satisfied then the operator $\pi_0^{\Q}$ (or $\tilde{\pi}_0^{\Q}$) is no longer 
an orthogonal projection, and 
equation \eqref{PiQ*} does not hold. It is rather difficult  
to obtain an explicit expression for $\pi_0^{\Q}$ in this case. 
We also remark that estimating the convergence 
constant $\alpha_\Q$ in \eqref{QKQ_estimation1} 
is a non-trivial task since such constant coincides 
with the real part of the smallest non-zero 
eigenvalue of $\Q\K\Q$.

\subsection{EMZ memory and fluctuation terms}
Proposition \ref{prop2} allows us to prove that 
the EMZ memory kernel \eqref{SFD} and the fluctuation 
term \eqref{f} of the particle system converge 
exponentially fast to an equilibrium state for any 
observable \eqref{observable}. 

\begin{corollary} \label{cor:K_conver1}
Under the same hypotheses of Proposition \ref{prop2} and Corollary \ref{cor:K_conver} the one-dimensional memory kernel 
$K(t)= \langle u(0),\K e^{t\Q\K\Q}\Q\K u(0)\rangle_{\rho_{eq}}/\langle u(0)^2\rangle_{\rho_{eq}}$
converges to an equilibrium state exponentially fast in time, i.e., 
\begin{align}\label{K(t)_Langevin}
\left|K(t)-\left(\langle\K u_0\rangle_{\rho_{eq}} \langle\Q\K^*u_0\rangle_{\rho_{eq}} 
+\langle\Q\K^* u_0,w\rangle_{\rho_{eq}} \langle\K u_0,w\rangle_{\rho_{eq}} 
\right)\right|
\leq Ce^{-\alpha_{\Q}t},
\end{align}
where $\K w=\K^* w=u$.
\end{corollary} 
\begin{proof}
The Corollary follows immediately from \eqref{kernel_est},  
\eqref{PiQ*} and the fact that $\P\Q=0$. 

\end{proof}

\noindent 
We emphasize that if $w$ is known then the equilibrium state 
can be calculated explicitly.  
It is straightforward to extend \eqref{K(t)_Langevin}
to matrix-valued memory kernels \eqref{SFD}.
By following the same steps that lead us to 
\eqref{matrixK}, we obtain 
\begin{align}\label{matrixK1}
\|\bm K(t)-\bm G^{-1}\bm C^{\Q}\|_{\mathscr{M}}\leq C\|\bm G^{-1}\bm D^{\Q}\|_{\mathscr{M}}e^{-\alpha_{\Q}t}, 
\end{align}
where $\|\cdot\|_{\mathscr{M}}$ denotes any matrix norm, 
and $\bm G$ is the Gram matrix \eqref{gram}. 
The entries of the matrix $\bm D^{\Q}$ and $\bm C^{\Q}$ are given 
explicitly by 
\begin{align*}
D^{\Q}_{ij}&=\|\Q\K^*u_i(0)\|_{L^2_{eq}}
\|\K u_j(0)\|_{L^2_{eq}},\\
C^{\Q}_{ij}&=\langle\Q\K^*u_i(0)\rangle_{\rho_{eq}} \langle\K u_j(0)\rangle_{\rho_{eq}} +\sum_{k=1}^m\langle\K u_j(0),w_k(0)\rangle_{\rho_{eq}} \langle\Q\K^* u_i(0),w_k(0)\rangle_{\rho_{eq}} .
\end{align*}
The components of the EMZ fluctuation term \eqref{f}
decay to an equilibrium state as well, exponentially fast 
in time. In fact, if we choose the initial condition as 
$\bm u_0=\Q\K \bm u_0$, then \eqref{QKQ_estimation} 
yields the following $L^2(\R^n;\rho_{eq})$-equivalent 
estimate 
\begin{align}\label{QKQ_estimation_forcei}
\Bigl\|f_j(t)-\Bigl(\langle\Q\K u_j(0)\rangle_{\rho_{eq}} +\sum_{k=1}^m\langle\Q\K u_j(0),w_k(0)\rangle_{\rho_{eq}} \Bigr)\Bigr\|_{L^2_{eq}}
\leq Ce^{-\alpha_{\Q} t}\|\Q\K u_j(0)\|_{L^2_{eq}}.
\end{align}
The inequality \eqref{QKQ_estimation_forcei} can be 
written in a vector form as
\begin{align}
\label{QKQ_estimation_forcei1}
\Bigl\|\bm f(t)-
\Bigl(\langle\Q\K \bm u(0)\rangle_{\rho_{eq}} +\sum_{k=1}^m&\langle\Q\K\bm u(0),w_k(0)\rangle_{\rho_{eq}} w_k(0)\Bigr)\Bigr\|_{V_{eq}}\nonumber \\
&\leq Ce^{-\alpha_{\Q}t}\| 
(\|\Q\K u_1(0)\|_{L^2_{eq}},\cdots,\|\Q\K u_m(0)\|_{L^2_{eq}}
)\|_\mathscr{M},
\end{align}
where $\|\cdot\|_{V_{eq}}$ is a norm in the tensor product 
space $V_{eq}=\otimes_{i=1}^m L^2(\R^n;\rho_{eq})$, defined similarly to \eqref{prod_norm}.

\section{Summary}
\label{sec:conclusion}

We developed a thorough mathematical analysis of the 
effective Mori-Zwanzig equation governing the dynamics of 
noise-averaged observables in nonlinear dynamical 
systems driven by multiplicative Gaussian white noise.   
Building upon recent work of Eckmann, Hairer, Helffer, H\'erau and Nier  
\cite{eckmann2003spectral, nier2005hypoelliptic} 
on the spectral properties of hypoelliptic operators, 
we proved that the EMZ memory kernel 
and fluctuation terms converge exponentially fast (in time) 
to a {computable} equilibrium state. 
This allows us to effectively study the asymptotic 
dynamics of any smooth quantity of interest 
depending on the stochastic flow generated by the SDE \eqref{eqn:sde}. 
We applied our theoretical results to 
a particle system widely used in statistical mechanics
to model the mesoscale dynamics of liquids and gasses, 
and proved that for smooth polynomial-bounded 
particle interaction potentials the EMZ memory 
and fluctuation terms decay exponentially fast 
in time to a unique equilibrium state.
Such an equilibrium state depends on the kernel of 
the orthogonal dynamics generator $\Q\K\Q$ 
and its adjoint $\Q\K^*\Q$.
We conclude by emphasizing that the Mori-Zwanzig 
framework we developed in this paper can be 
generalized to other stochastic dynamical 
systems, e.g., systems driven by fractional 
Brownian motion with anomalous long-time behavior
\cite{bazzani2003diffusion,denisov2009generalized,
morita1980contraction}, provided there exists a 
strongly continuous semigroup for such systems  
that characterizes the dynamics of noise-averaged observables.

\vspace{0.3cm}
\noindent 
{\bf Acknowledgements} 
This research was partially supported by the Air Force Office 
of Scientific Research (AFOSR) grant FA9550-16-586-1-0092 and by the National Science Foundation (NSF) grant 2023495 -- TRIPODS: Institute for Foundations of Data Science. The authors would like to thank 
Prof. F. H\'erau, Prof. B. Helffer and Prof. F. Nier for helpful discussions 
on the spectral properties of 
the Kolmogorov operator.  

\vspace{0.3cm}
\noindent 
{\bf Data availability statement} 
The data that support the findings of this study 
are available from the corresponding author upon 
request.

\vspace{0.3cm}
\noindent 
{\bf Conflict of interest} 
The authors have no conflicts to disclose.

\bibliographystyle{plain}
\bibliography{4}
\end{document}